\documentclass[journal,twocolumn]{IEEEtran}%
\usepackage{cite}
\usepackage{amsmath,amssymb,amsfonts}
\usepackage{algorithm, algpseudocode} 
\usepackage{subcaption} 
\usepackage{graphicx}
\usepackage{textcomp}
\usepackage{xcolor}
\usepackage{xfrac}
\usepackage{amsmath}
\usepackage{aliascnt}
\usepackage{amsfonts}
\usepackage{amssymb}%
\usepackage{todonotes}
\usepackage{enumerate}
\usepackage{comment}
\usepackage{tikz}
\usetikzlibrary{arrows.meta, positioning, fit, calc}

\providecommand{\U}[1]{\protect\rule{.1in}{.1in}}
\def\BibTeX{{\rm B\kern-.05em{\sc i\kern-.025em b}\kern-.08em
T\kern-.1667em\lower.7ex\hbox{E}\kern-.125emX}}

\usepackage[utf8]{inputenc}
\usepackage{cite}
\usepackage{hyperref}
\usepackage{amsmath,amssymb,amsfonts,amsthm}
\usepackage{graphicx}
\usepackage{textcomp}
\usepackage{cancel}
\usepackage{xcolor}
\def\BibTeX{{\rm B\kern-.05em{\sc i\kern-.025em b}\kern-.08em
    T\kern-.1667em\lower.7ex\hbox{E}\kern-.125emX}}
\usepackage{xfrac}
\providecommand{\U}[1]{\protect\rule{.1in}{.1in}}

\newcommand{\C}{\mathbb{C}}

\newcommand{\R}{\mathbb{R}}

\newcommand{\Diag}{\operatorname{Diag}} 
\newcommand{\Cov}{\operatorname{Cov}} 

\newcommand{\spann}{\mathrm{span}}

\makeatletter
\newcommand{\Autoref}[1]{%
  \begingroup
  \def\sectionautorefname{Section}%
  \def\subsectionautorefname{Subsection}%
  \def\subsubsectionautorefname{Subsection}%
  \def\figureautorefname{Figure}%
  \def\tableautorefname{Table}%
  \def\algorithmautorefname{Algorithm}%
  \autoref{#1}%
  \endgroup
}
\makeatother

\usepackage{tikz}
\usetikzlibrary{
    positioning,
    fit,
    arrows.meta,
    shapes.geometric,
    shapes.misc,
    calc
}

\newcommand{\calD}{\mathcal{D}}

\newcommand{\calF}{\mathcal{F}}

\newcommand{\calH}{\mathcal{H}}

\newcommand{\calL}{\mathcal{L}}

\newcommand{\calN}{\mathcal{N}}

\newcommand{\ralpha}{\mathrm{r}\alpha}
\newcommand{\ialpha}{\mathrm{i}\alpha}

\newcommand{\abs}[1]{\left\vert #1 \right\vert}

\newcommand{\sse}{\subseteq}
\newcommand{\sprod}[1]{\left\langle #1 \right\rangle}

\usepackage{tikz}
\usetikzlibrary{cd}

\usepackage{dsfont}

\DeclareMathOperator{\supp}{supp}

\usepackage{dsfont}

\renewcommand{\Re}{\mathrm{Re}}

\newtheorem{lem}{Lemma}

\newaliascnt{prop}{lem}
\newtheorem{prop}[prop]{Proposition}
\aliascntresetthe{prop}

\newaliascnt{theo}{lem}
\newtheorem{theo}[theo]{Theorem}
\aliascntresetthe{theo}

\newaliascnt{cor}{lem}

\aliascntresetthe{cor}

\newaliascnt{assumption}{lem}

\aliascntresetthe{assumption}

\newaliascnt{defi}{lem}
\newtheorem{defi}[defi]{Definition}
\aliascntresetthe{defi}

\newaliascnt{rem}{lem}
\newtheorem{rem}[rem]{Remark}
\aliascntresetthe{rem}

\newaliascnt{res}{lem}

\aliascntresetthe{res}

\newaliascnt{problem}{lem}

\aliascntresetthe{problem}

\newaliascnt{ass}{lem}
\newtheorem{ass}[ass]{Attack model}
\aliascntresetthe{ass}

\definecolor{gerhard}{rgb}{1,0,0}

\begin{document}
\title{Bilinear Compressive Security with Finitely Many Samples}

\author{Axel Flinth,
Semira Einsele,
Matilda Kylberg,
Hubert Orlicki,
Gerhard Wunder.
\thanks{A. Flinth and M. Kylberg are with Umeå University, Sweden}
\thanks{H. Orlicki is with Cracow University of Technology, Poland}
\thanks{S. Einsele and G. Wunder are with Freie Universität Berlin, Germany}
\thanks{Part of this work was presented at the IEEE International Conference on Communications (ICC) 2026 \cite{flinth2026bilinear}, and in Kylberg's Master's thesis \cite{matilda_thesis}.}
}

\maketitle

\begin{abstract}
    Bilinear Compressive Security (BCS) is a lightweight encryption scheme that combines compressed sensing with a fresh random convolution filter to improve resistance against key-recovery attacks under key reuse. Previous work derived information-theoretic guarantees for BCS by relating key identifiability to phase retrieval under the assumption that an adversary has access to the full ciphertext distribution induced by the random filter. In practice, however, an adversary only observes finitely many encryptions, leaving the security of BCS in the finite-sample regime largely unexplored.

    In this paper, we study key recovery from only finitely many observations. We first develop a covariance-based attack that estimates the linear measurements from repeated encryptions of the same plaintext and subsequently recovers the secret measurement matrix. We analyze the estimation procedure and derive finite-sample guarantees together with corresponding sample-complexity bounds. We further investigate the complete attack experimentally and compare its observed recovery behavior with the previously established information-theoretic identifiability results.

    Our experiments show two qualitatively different recovery regimes on the structure of the measurement matrix. For dense matrices, including Gaussian matrices, the number of measurements required for successful recovery appears to grow linearly with the signal dimension. This differs from the worst-case behavior suggested by information theory. In contrast, for matrices with sparse rows, we observe a quadratic growth consistent with the theoretical bound. 

    These results suggest that the sparsity structure of the measurement matrix plays an important role in the security of BCS and should be taken into account when selecting parameters and designing practical instances of the scheme.
\end{abstract}

\section{Introduction}

Compressed sensing \cite{candes2005decoding,FouRau2013} enables the recovery of a sparse signal $x \in \mathbb{C}^n$ from a reduced number of linear measurements
$$
z = Qx \in \mathbb{C}^m, \qquad m<n,
$$
provided that the measurement matrix $Q \in \mathbb{C}^{m \times n}$ satisfies suitable structural conditions, such as the \emph{restricted isometry property} \cite{candes:rip2008}. 
In addition to reducing the dimensionality of the transmitted data, this observation has naturally motivated lightweight communication schemes in which the measurement matrix $Q$ is shared as a secret key: Alice transmits $Qx$, while Bob reconstructs the original $x$ using a compressed-sensing recovery algorithm. Such constructions are attractive in resource-constrained communication settings because compression and encryption are performed simultaneously.

By letting $Q$ be unknown to an eavesdropper, this communication scheme naturally offers a layer of security  \cite{rachlin2008secrecy,testa2019compressed}. The security of this linear construction, however, deteriorates when the same secret measurement matrix is reused\cite{bianchi2014security,Cambaberi2015}. If an adversary Eve obtains sufficiently many plaintext--ciphertext pairs $(x_k, Qx_k)$, then the secret matrix $Q$ can be recovered by solving a linear system. In particular, $n$ linearly independent plaintexts uniquely determine a measurement matrix with $n$ columns. Consequently, without additional randomization, key reuse provides no information-theoretic protection against sufficiently strong known- or chosen-plaintext attacks.

\subsection{Bilinear compressive security}

Bilinear Compressive Security (BCS), introduced by a subset of this article's authors in \cite{flinth2026bilinear}, was proposed to mitigate this key-reuse vulnerability by introducing fresh randomness into every encryption. Instead of directly transmitting the linear measurement $Qx$, Alice samples a fresh random convolution filter $h \in \mathbb{C}^m$ and sends
$$
y = h *Qx.
$$
The distribution of $h$ is public, while its concrete realization is neither transmitted nor stored. Bob, who knows $Q$ but not $h$, recovers the plaintext by solving a sparse blind deconvolution problem. Under suitable assumptions on the sparsity of the message $x$ and the filter $h$, existing efficient blind deconvolution recovery algorithms such as HiHTP \cite{hiHTP} and sparse power factorization \cite{flinth2022bisparse, lee2017near} allow efficient recovery with high probability.

The random filter hides the linear measurement $Qx$ from the adversary and changes the key-recovery problem. In the mentioned \cite{flinth2026bilinear}, the identifiability of the secret codebook $Q$ and its connection to phase retrieval was studied. In particular, it  established conditions under which different codebooks cannot be distinguished, even if Eve has access to the full ciphertext distributions. However, it leaves the problem whether Eve can recover the codebook when only finitely many ciphertexts are available open.

The use of fresh randomness to prevent key recovery under repeated observations is also familiar from cryptographic constructions based on the Learning Parity with Noise (LPN) problem. In LPN, randomness enters additively through Bernoulli noise, whereas in BCS it enters multiplicatively through the random convolution filter. We mention this analogy only at a conceptual level; the resulting recovery problems are otherwise different.

\subsection{Injectivity and practical recoverability in the  infinite sample regime}
As mentioned, a preliminary version of BCS together with its information-theoretic security analysis was presented by a subset of the authors in  \cite{flinth2026bilinear}. There, the security of BCS was studied in an idealized known- and chosen plaintext setting in which the adversary has access to perfect knowledge of the  ciphertext \emph{distributions} induced by the random filter conditioned on different plaintexts $x_k$. In other words, the eavesdropper is assumed to have access to an infinite amount of samples $h_p*Qx_k$, $p =1, 2, \dots$ for each $x_k$. An injectivity result was shown:  Under the weak assumption that the distribution of $h$ is phase-symmetric, there exists two non-equal codebooks that cannot be distinguished from the $y$-distributions if and only if the vector $q$ cannot be determined from the knowledge of the amplitude measurements $|\langle q,x_k\rangle|$. That is, if and only if the $x_k$ \emph{do not do phase retrieval} \cite{bandeira2014saving}. Using this connection, it was shown that when the messages $x_k$ are $s$-sparse, recovery of $Q$ in general necessitates observation from $M \geq \max(n,n^2/s^2)$, rather than the $n$ needed for ordinary, linear,  compressed sensing.

Crucially, the result shown in \cite{flinth2026bilinear} is an injectivity result: It only gives a condition when there exists two $Q,Q'$ that are confounded by the BCS protocol -- it does not give a condition that implies that every $Q$ is confounded. In this work, we will therefore perform a more thorough numerical study of the relevance of the injectivity result for practical recovery. 

\subsection{A practical attack for the finite sample regime}
Another limitation of the analysis from \cite{flinth2026bilinear} is the assumption that the adversary is assumed to be given perfect knowledge of the  ciphertext \emph{distributions} directly. In practice, they instead only have access to finitely many samples from the distributions of ciphertexts from each plaintext. 
The present paper intents to  bridge this gap, by proposing and analyzing a practical attack 
only needing access finitely many samples of the $y_k$. We will also describe designs of the filter and codebook in which the distribution can be recovered from a single sample, and the finite sampling regime hence provides no extra layer of security.

\subsection{Contributions}
Our main contributions are as follows. 

\begin{itemize}
    \item We develop a practical covariance-based key-recovery attack for BCS in the finite-sample setting. The attack estimates the hidden measurements $Qx_k$ \emph{projectively} (i.e. modulo unit scalars) from repeated ciphertexts, and subsequently reconstructs the secret measurement matrix. We give sample complexity bounds and describe failure cases.

    \item We give a complementary single-sample result showing that, for particular structured choices of the codebook and the filter, a finite sample setting provides no additional level of security compared to the infinite-sample setting.
    
      \item We perform an extended numerical study to quantify the gap between worst-case identifiability and recovery of concrete codebooks. We find that dense codebooks exhibit an approximately linear recovery threshold in the signal dimension, whereas codebooks with sparse rows exhibit quadratic $(n/s)^2$-type behavior suggested by the information-theoretic bound.
    
    \item Finally, we complete the preliminary conference version by providing full proofs of the information-theoretic identifiability results.
\end{itemize}

Our results separate two distinct aspects of key recovery in BCS: the identifiability of the codebook from the projective measurements $\Phi([Q])$, and the statistical difficulty of estimating these measurements from finitely many ciphertext observations. Our experiments further show that practical codebook recovery depends strongly on the structure of the sensing matrix, while the feasibility and sample complexity of our finite-sample attack depend on the distribution of the random filter.

The remainder of this paper is organized as follows. \Autoref{sec:system_model} introduces the system model and adversarial setting.  \Autoref{sec:security_analysis} revisits the information-theoretic identifiability analysis, and \Autoref{sec:proof} provides the complete proof of our guarantees. \Autoref{sec:finite_samples} presents the finite-sample recovery attack together with its theoretical analysis. Experimental results are reported in \Autoref{sec:numerics}, followed by a discussion and conclusion.

\section{System model} \label{sec:system_model}
We begin by giving a formal definition of the BCS scheme, and subsequently discuss its correctness, and the adversarial observation model.

\subsection{Notation}
Vectors and matrices are defined over $\mathbb{C}$ unless stated otherwise. The Euclidean and Frobenius norm are denoted by $\|\cdot\|_2$ and $\|\cdot\|_{\mathrm F}$, respectively. The support and sparsity of a vector $v$ are denoted by
$$
\supp(v)=\{i:v_i\neq 0\},
    \qquad
    \|v\|_0=|\supp(v)|.
$$
A vector is called $s$-sparse when $\|v\|_0\leq s$.

For a matrix $A$, we write $A_{i,:}$ and $A_{:,j}$ for its $i$th row and $j$th column, respectively. The circular convolution of two vectors is denoted $*$, and the discrete Fourier transform by $\mathcal{F}.$ The Fourier transform of a vector $x$ is written as $\widehat{x} = \mathcal{F}x$. For convenience, we will define the Fourier transformation of a matrix $A$ differently:  $(\mathcal{F}A)\widehat{x}= \widehat{Ax}$ , $x\in \C^n$. 

Let $\mathbb S^1 := \{\theta\in\mathbb C : |\theta|=1\}$ denote the unit circle in $\C$. For $x,y \in \C^d$, let 
$$
x \sim y \quad \Longleftrightarrow \quad x=\theta y \text{ for some }\theta\in\mathbb S^1.
$$
We denote the corresponding quotient space
$$
\mathbb C^d/\mathbb S^1,
$$
and write $[x]\in\mathbb{C}^d/\mathbb{S}^1$ for the equivalence class of $x$ modulo a global phase.

\subsection{Bilinear Compressive Security}
The Bilinear Compressive Security (BCS) scheme, introduced in \cite{flinth2026bilinear}, combines compressed sensing with sparse blind deconvolution to provide lightweight encryption that remains resistant to key-recovery attacks under key reuse. The scheme is parameterized by
\begin{itemize}
    \item[(i)] a secret measurement matrix (\emph{codebook})
    $$
    Q \in \mathbb{C}^{m \times n},
    $$ shared between Alice and Bob, and
    \item[(ii)] a public probability distribution
    $
    \mathcal{D}
    $ of convolution filters
    $$
    h \in \mathbb{C}^m.
    $$
\end{itemize}

The concrete realization of the filter is generated independently for every encryption and is never transmitted nor stored. The BCS scheme consists of the following algorithms.

\paragraph{Key generation}
Alice and Bob agree on a secret measurement matrix
$$
Q\in\mathbb{C}^{m\times n}.
$$
The distribution $\mathcal D$ from which random filters are drawn is publicly known.

\paragraph{Encryption}
To encrypt a plaintext
$
x\in\mathbb{C}^n$,
Alice first computes the compressed representation
$$w=Qx.$$

She then samples an independent random filter
$
h\sim\mathcal D
$ and
forms the ciphertext
$$
y=h*w = h*Qx.
$$

The filter realization is discarded immediately after encryption, so encrypting the same plaintext twice produces different ciphertexts with overwhelming probability.

\paragraph{Decryption}
Bob receives the ciphertext $y$ and knows the sensing matrix $Q$, but not the realization of the random filter. Recovering the plaintext therefore amounts to solving the blind deconvolution problem
$$
y=h*Qx.
$$

We discuss the corresponding recovery guarantees in the following subsection.

\subsection{Correctness}
The correctness of the BCS scheme follows from existing recovery guarantees for sparse blind deconvolution. Recall that Bob observes
$$
    y = h * Qx
$$
and knows $Q$, but not the realization of the random filter $h$.

As is standard in blind deconvolution, $(h,x)$ is identifiable only up to a nonzero scaling factor, since
$$
    (\alpha h) * Q(\alpha^{-1}x) = h * Qx
$$
for every $\alpha \in \mathbb{C}\setminus\{0\}$. Under suitable assumptions on $Q$, $x$, and $h$, this is the only ambiguity. 

Several efficient algorithms have been proposed for sparse blind deconvolution, including convex optimization approaches \cite{convex_blind_deconv}, alternating minimization \cite{alternating}, Sparse Power Factorization \cite{lee2017near}, and HiHTP \cite{hiHTP,flinth2022bisparse}.

The following theorem summarizes the recovery guarantees most relevant to BCS.

\begin{theo}[\cite{flinth2022bisparse,lee2017near}]
Let $Q\in\mathbb C^{m\times n}$ be a random matrix with independent Gaussian entries, and suppose that the plaintext $x$ is $s$-sparse and the filter $h$ is $\sigma$-sparse. Then there exists a constant
$C>0$ such that if
$$
m\ge
C\,s\sigma\,\mathrm{polylog}(m,n),
$$
algorithms such as HiHTP and Sparse Power Factorization recover the pair $(h,x)$, up to the inherent scaling ambiguity, with high probability over the draw of $Q$.
If, in addition, both $h$ and $x$ satisfy suitable spectral flatness conditions, Sparse Power Factorization requires only
$$ m\ge C\,(s+\sigma)\,\mathrm{polylog}(m,n)$$
measurements.
\end{theo}

\begin{figure}
    \centering
    \includegraphics[width=0.9\linewidth]{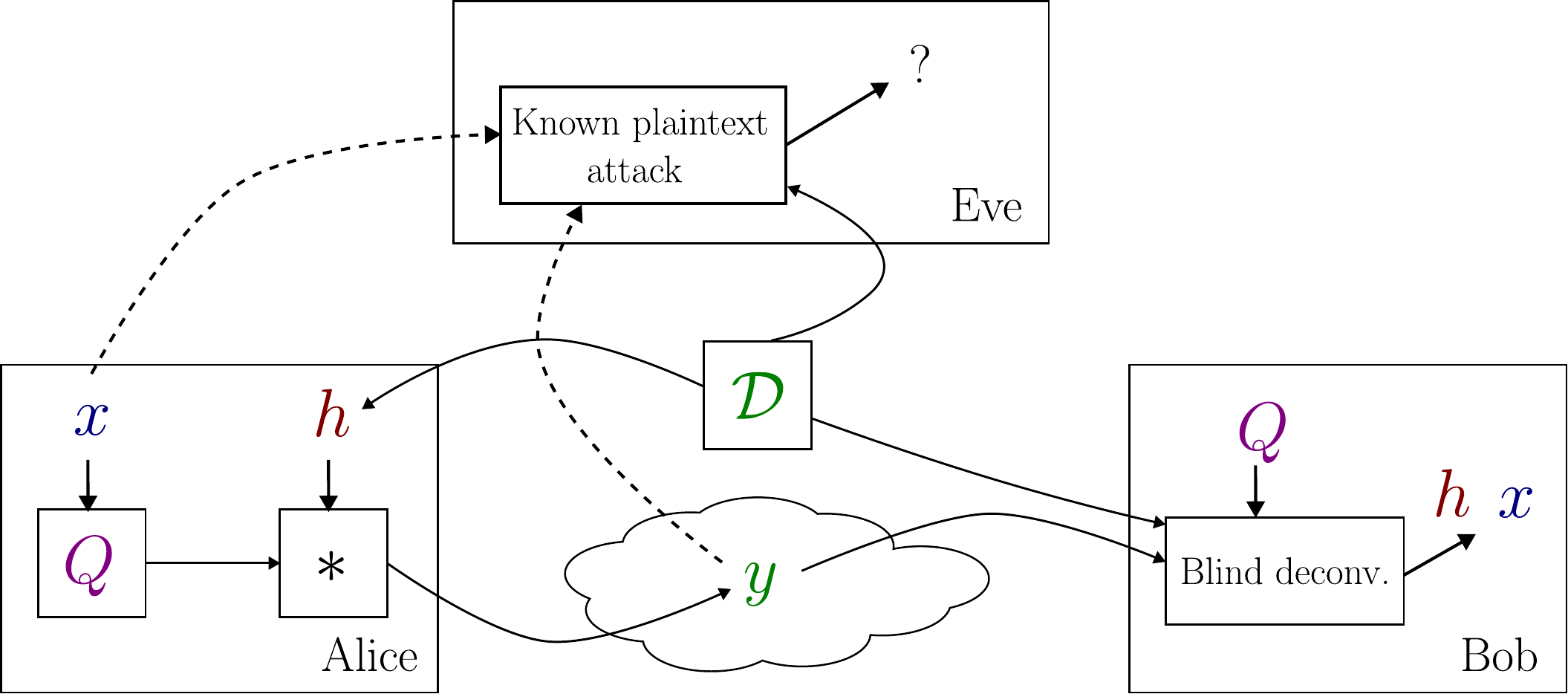}
    \caption{A graphical depiction of the BCS scheme, and the considered known plaintext attack. }
    \label{fig:placeholder}
\end{figure}

Throughout this paper we assume that the parameters of the BCS scheme are chosen such that decryption succeeds with overwhelming probability. The particular choice of the blind deconvolution algorithm is not important for the subsequent security analysis; we only require that legitimate users can reliably recover the transmitted plaintext under the stated assumptions. For further details on sparse blind deconvolution and its recovery guarantees we refer the reader to \cite{flinth2022bisparse,matilda_thesis}.

\subsection{Complexity} 
On the encryption side, BCS introduces only the additional cost of sampling a fresh filter $h$ and computing the convolution
$$
y= h*Qx
$$
after the linear measurement $Qx.$ The filter does not need to be stored or transmitted and the convolution is simple to implement. 

Decryption is computationally more demanding, since Bob must solve a blind deconvolution problem, and at the same time recover the unknown filter and plaintext. Its exact computational cost depends on the chosen recovery algorithm. However, there do exist numerically efficient algorithms to perform the blind deconvolutions. The sparse power factorization essentially boils down to repeatedly applying operations used for linear recovery, and hence has essentially the same cost.

One has to take into account that the measurement matrix $Q$ needs to have more rows in order for blind recovery to succeed compared to solving the standard linear compressed sensing problem $Qx\to x$. However, for a fixed filter-sparsity $\sigma$, the number of required measurements still grows linearly with the number of non-zero elements in $x$. Thus, for fixed filter sparsity $\sigma$, the dependence on the plaintext sparsity s remains linear, up to logarithmic factors, as in standard compressed sensing.

\subsection{Adversarial Model}\label{sec:adverserial-model}
Throughout this paper, we consider a strong known-plaintext or chosen-plaintext adversary, denoted by Eve. The adversary knows the public parameters of the scheme, including the distribution $\mathcal{D}$ of the random convolution filters, but neither the secret sensing matrix
$$
Q \in \mathbb{C}^{m \times n}$$
nor the individual filter realizations used during encryption.

Eve observes plaintext--ciphertext pairs

$$
(x_k,y_{k,p}), \qquad k\in \{1, \dots, M\},\; p\in \{1, \dots, r\},
$$
where
$$
y_{k,p} = h_{k,p}*Qx_k,
$$
and each filter
$$
h_{k,p} \sim \mathcal{D}
$$
is sampled independently. Here, $M$ denotes the number of distinct plaintexts available to Eve, and $r$ the number of independent encryptions observed for each plaintext. In the known-plaintext setting, these observations arise from legitimate communications, whereas in the chosen-plaintext setting Eve may choose the plaintexts $x_k$ and obtain repeated independent encryptions of each.

Eve's objective is to recover the secret sensing matrix $Q$. Knowledge of $Q$ allows Eve to apply the same blind deconvolution procedure as a legitimate receiver to future ciphertexts. Thus, key recovery compromises the confidentiality of subsequent encryptions under the same measurement matrix.

Our key-recovery attack naturally decomposes into two stages. 

First, for each plaintext $x_k$, Eve uses the repeated ciphertext observations $\{y_{k,p}\}_{p=1}^r$ to estimate the corresponding projective codeword $[Qx_k]$ by exploiting the statistical properties of the random filters. Second, the resulting projective measurements
$$
[Qx_1],\ldots,[Qx_M]
$$
are used to reconstruct the secret sensing matrix $Q$. The first stage is statistical and constitutes the main novel contribution beyond the information-theoretic setting of~\cite{flinth2026bilinear}, whereas the second relies on the identifiability results established in the next section.

\section{Security analysis} \label{sec:security_analysis}
In this section, we revisit the information-theoretic security analysis of Bilinear Compressive Security. Unlike the finite-sample setting considered later in the paper, we assume here that Eve has access to the exact ciphertext distribution induced by the random filters. This idealized observation model allows us to characterize when the secret measurement matrix is uniquely determined by the available observations. The resulting identifiability results form the theoretical foundation for the finite-sample attack developed in \Autoref{sec:finite_samples}.

\subsection{Attack model}
As discussed in \Autoref{sec:adverserial-model}, Eve observes encryptions 
$$
y_{k,p}=h_{k,p}*Qx_k, \qquad h_{k,j}\sim\mathcal D,
$$
for a collection of plaintexts $x_k$, $k\in \{1, \dots, M\}$. In this section, we consider the idealized infinite-sample limit of this adversarial model: for each plaintext $x_k$, Eve is assumed to know the exact ciphertext distribution induced by the random filter.

This idealization removes the finite-sample estimation problem and allows us to focus solely on the information-theoretic identifiability: 
\begin{quote}
\emph{Under what conditions on $\mathcal D$ and $(x_k)_{k\in \{1, \dots, M\}}$ do these ciphertext distributions uniquely determine the secret measurement matrix $Q$?}
\end{quote}

We formalize this idealized observation model as follows:

\begin{ass}[Infinite-sample observation model]
Let $x_1,\ldots,x_M$ be plaintexts available to Eve. For each $x_k$, Eve has access to the exact probability distribution of 
$$
h*Qx_k \qquad h\sim\mathcal D.
$$ 
Eve has complete knowledge of the filter distribution $\mathcal D$, and can observe as many ciphertexts $y_{k,p},~p=0,1,2, \dots$ as needed, but does not observe the individual filter realizations.
\end{ass}

This model is highly favorable to Eve. In particular, exact knowledge of the ciphertext distributions corresponds to the idealized limit of arbitrarily many independent encryptions of each plaintext. In \Autoref{sec:finite_samples}, we return to the practical setting in which Eve observes only finitely many ciphertexts per plaintext.

\subsection{Preliminaries}
What should we consider a successful recovery of the secret key $Q$? We first formalize what it means to recover the measurement matrix $Q$. Note, since here is a scaling ambiguity, the ciphertext distributions considered below may be invariant under multiplication of $Q$ by a global unit scalar, we consider key rocovery up to such a global phase. Recall that $\mathbb{S}^1=\{\theta\in\mathbb C:|\theta|=1\}$.

\begin{defi}[Key retrieval] \label{def:rec}
Let $h$ have a known distribution $\mathcal{D}$. We say that a set of vectors $(x_k)_{k\in \{1, \dots, M\}}$ can \emph{retrieve the key} for $\mathcal D$-distributed $h$ if
the map 
$$
[Q]\mapsto (h*Qx_k)_{k \in \{1, \dots, M\}}
$$ 
from $\C^{m \times n}/ \mathbb{S}^1$ to the set of probability distributions on $(\C^m)^M$ is injective, where the right-hand side denotes the probability distribution induced by the draw $h \sim \mathcal{D}.$
\end{defi}

\begin{rem}
    The intuitive meaning of \autoref{def:rec} is most easily explained in the negative: 
    Key retrieval fails, if there exist two keys $Q, Q'$ that are not unit scalar multiples of each other but for which $h*Qx_k$ and $h*Q'x_k$ have the same distribution for every $k \in \{1, \dots, M\}$. 
\end{rem}

Before we get to the actual assumption on $\calD$ that we wish to study, we shortly discuss, as in \cite{flinth2026bilinear}, the case where $h$ is standard complex Gaussian with covariance matrix $I$. This example illustrates the main ideas of the analysis.

By the Fourier convolution theorem, the distribution of $h*Qx_k$ is equivalently described in the Fourier domain by

\begin{align*}
    \widehat{h} \odot \widehat{Qx_k},
\end{align*}
where $\odot$ denotes the pointwise product. Since $h$ is standard complex Gaussian, so is its Fourier transform $\widehat{h}$. Consequently, its entries are independent and rotationally symmetric, so the ciphertext distributions will never reveal more than the \emph{magnitude vectors} $|FQx_k|$, $k \in \{1, \dots, M\}$, where $F$ denotes the Fourier matrix. 

Since $F$ is invertible, recovering $Q$ from these magnitude vectors is equivalent to recovering $N=FQ$ from
$$
|Nx_k|, \qquad k \in \{1, \dots, M\}.
$$

Let $n_\ell$ denote the $\ell$th row of $N$. These observations are precisely the magnitude measurements 

\begin{align} \label{eq:indep_PR}
 \left|\langle n_\ell,x_k\rangle\right|, \quad {\ell \in \{1, \dots, m\}, \; k \in \{1, \dots, M\}}.
\end{align}

However, these measurements can never determine $N$ up to a single global phase: Multiplying each row $n_\ell$ with different unit scalars $\theta_\ell \in \mathbb{S}^1$ leaves all measurements in \eqref{eq:indep_PR} unchanged. Thus:

\begin{prop} \label{prop:gaussian}
    No set of vectors can retrieve the key for standard complex Gaussian distributed $h$.
\end{prop}

\begin{rem}
    We stress that the above proposition is only of illustrative relevance. Indeed, when $h$ is Gaussian distributed, it will almost surely have all entries non-zero, making the discussed sparse blind deconvolution procedures inapplicable.
\end{rem}

Now let us now take a step back and inspect \eqref{eq:indep_PR}. Mathematically, it is a system of \emph{uncoupled phase retrieval problems}.

\begin{defi}[Phase retrieval]
    Let $(x_k)_{k\in \{1, \dots, M\}}$ be a set of vectors in $\C^n$. The problem of recovering a vector $q \in \C^n$ up to a global phase from the measurements 
    $$
    \bigl(|\langle q,x_k\rangle|\bigr)_{k\in\{1, \dots, M\}}
    $$ 
is called the \emph{phase retrieval problem}. We say that vectors $(x_k)_{k \in \{1, \dots, M\}}$ \emph{do phase retrieval} if the map
\begin{align*}
\phi: \C^n/\mathbb{S}^1 &\to \R^M, \\
[q] &\mapsto \bigl(|\langle q,x_k\rangle|\bigr)_{k\in\{1, \dots, M\}}
\end{align*}
is injective.
\end{defi}

There are many sets of vectors that do phase retrieval. In fact, a generic set of $M$ vectors in $\C^n$ with $M\geq 4n-4$ does  \cite{bandeira2014saving}. The reason \eqref{eq:indep_PR} cannot be solved is hence not that we are dealing with a set of phase retrieval problems, but rather that they are \emph{uncoupled}. Indeed, a set of vectors that do phase retrieval will be able to recover each row $n_\ell$ of $N$ up to a unit scalar $\theta_\ell$, but since these unit scalars may be different, that will not reveal $N$ up to one global unit scalar.

\subsection{Phase symmetric filters}
As already discussed, the assumption that the filters are standard complex Gaussian is too powerful to allow for codebook retrieval: the independence of the entries of their Fourier transforms leads to the uncoupled phase-retrieval problems in \eqref{eq:indep_PR}. On the other hand, some randomness in the filter distribution $\calD$ is necessary for BCS to provide additional security over the linear compressed-sensing case. Namely, if $\calD$ is deterministic and, for example, always outputs the trivial filter $e_0$, then
$$
h*Qx_k=e_0*Qx_k=Qx_k,
$$
and the problem reduces to the linear case.

It turns out that one only needs the following  non-restrictive assumption in order to obtain security:

\begin{defi}[Phase-symmetric filter distribution] A filter distribution $\calD$ is \emph{phase symmetric} if, for $h \sim \calD $ and for every unit scalar $\theta\in \mathbb{S}^1$, the random filters $h$ and $\theta h$ have the same distribution.
\end{defi}

This condition is fulfilled by many distributions. For instance, Alice may use the identity filter $e_0$ multiplied by a uniformly random unit scalar, leaving the vector $Qx$ minimally confounded. Another, more general, class of models that are phase symmetric is the following:
\begin{align} \label{eq:GaussianFill}
    h = z \odot b, \quad z\in \{0,1\}^n, \ b \sim \calN(0,\Sigma_b),
\end{align}
where a random support is first selected and its nonzero entries are then filled using entries of a possibly non-standard Gaussian vector.

Phase symmetry implies the following necessary condition for key retrieval: 

\begin{lem} \label{lem:PhiMap}
    If $h$ is phase symmetric, a set of vectors $(x_k)_{k \in \{1, \dots, M\}}$ can only retrieve the key if the map
     \begin{align*}
     \Phi: \mathbb C^{m\times n}/\mathbb S^1 &\to (\mathbb C^m/\mathbb S^1)^M, \\
    [Q]&\mapsto([Qx_k])_{k\in\{1, \dots, M\}}
    \end{align*}
    is injective.
\end{lem}

\begin{proof}
    Assume that $\Phi$ is not injective. Then there exist matrices $Q,Q'$ that are not unit scalar multiples of each other as well as unit scalars $\theta_k \in \mathbb{S}^1$ such that $$
    Qx_k=\theta_kQ'x_k \qquad \text{for every } k \in \{1, \dots, M\}.$$
Consequently,
    \begin{align*}
        h*Qx_k = h*(\theta_k Q'x_k)= (\theta_k h)*Q'x_k
    \end{align*}
    Now, since $h$ is phase symmetric, $\theta_k h$ and $h$ have the same distribution. Hence, the above equation shows that $h*Qx_k$ and $h*Q'x_k$ have the same distribution for every $k \in \{1, \dots, M\}$. Therefore, $Q$ and $Q'$ cannot be distinguished from ciphertext distributions, and the key cannot be retrieved.
\end{proof}

\begin{rem}
    It should be noted that \autoref{lem:PhiMap} is a one-way
statement: the injectivity of $\Phi$ is necessary but not sufficient for key retrieval. In particular, we do not claim that $\Phi([Q])$ can be determined from the ciphertext distributions induced by $h*Qx_k$, let alone from finitely many samples of these distributions.

Indeed, \autoref{prop:gaussian} gives a concrete and simple
counterexample to the converse of \autoref{lem:PhiMap}. A standard Gaussian filter $h$ is phase-symmetric, but the proposition shows that the key cannot be retrieved, irrespective of whether $\Phi$ is injective. 
\end{rem}

What is the mathematical structure of the problem of inverting the map $\Phi$? Recall that
$$
\Phi([Q])=\bigl([Qx_k]\bigr)_{k \in \{1, \dots, M\}},
$$
so the image of $[Q]$ consists of $M$ vectors $w_k=Qx_k$, each determined only up to an individual unit scalar. Equivalently, the amplitudes and relative phases for each $w_k$ are known, while its global phase remains undetermined.

Variants of relative phase retrieval, including practical recovery algorithms, were studied in \cite{kornprobst2021phase}. However, that work does not provide fundamental guarantees of injectivity of $\Phi$ required here. We establish the following Theorem, which is an interesting theoretical contribution of its own.

\begin{theo} \label{th:PRvsPhi}
     The map $\Phi$ from \autoref{lem:PhiMap} is injective if and only if the set of vectors $(x_k)_{k \in \{1, \dots, M\}}$ does phase retrieval.
\end{theo}
The proof of \autoref{th:PRvsPhi} is postponed to the next section as it is quite involved. Let us here instead use it to derive our fundamental security result. Combined with \autoref{lem:PhiMap}, the theorem shows that as long as the distribution of $h$ is guaranteed to be sparse, universal retrievability of the key will require observing ciphertexts for a number of plaintexts \emph{quadratic} in the dimension $n$, in contrast to the linearly many needed to break the linear compressed sensing scheme. 

The following characterization of phase retrieval will be useful for proving our main result.

\begin{lem}[{\cite[Lemma~9]{bandeira2014saving}}]\label{lem:phase-retrieval-characterization}
A set of vectors $(x_k)_{k \in \{1, \dots, M\}}$ does phase retrieval if and only if there is no nonzero Hermitian $H\in \C^{n \times n}$ of rank one or two such that
$$
\langle x_kx_k^*,H \rangle =0 \qquad \text{ for all } k\in \{1, \dots, M\}.
$$
\end{lem}
We can now state our main information-theoretic security result.

\begin{theo} \label{th:main} Let $s\leq n$ and let  $\calD$ be phase symmetric. Then no sets of $s$-sparse vectors $(x_k)_{k\in \{1, \dots, M\}}$ can retrieve the key 
\begin{enumerate}[(a)]
    \item for any $M$ if $s=1$,
    \item if $s\geq 2$ and 
    $$
    M< \max\big(\tfrac{n(n-1)}{s(s-1)}, 4n - 2 - 2\log_2(n-1)\big).
    $$
\end{enumerate}
\end{theo}

\begin{proof}
    By \autoref{th:PRvsPhi} and \autoref{lem:PhiMap}, it suffices to show that $(x_k)_{k\in \{1, \dots, M\}}$ does not do phase retrieval.
    
    Fundamental lower bound results on phase retrieval as given in Theorem 3 in \cite{heinosaari2013quantum}, show that phase retrieval is not possible when
    $$
    M< 4n-2-2\log_2(n-1),
    $$ which already proves one part of (b).

For the remaining claims, by \autoref{lem:phase-retrieval-characterization}, it suffices to construct a nonzero Hermitian matrix $H$ of rank one or two satisfying
$$
    \langle x_kx_k^*,H\rangle=0
    \qquad\text{for all } k\in\{1, \dots, M\}.
$$

   Now let us prove (a). Suppose that the $x_k$ are $1$-sparse. Then each matrix $x_kx_k^*$ is diagonal. 
   For any distinct $i,j \in \{1, \dots, n\}$, the matrix
   $$
   H=e_ie_j^*+e_je_i^*
   $$
   is Hermitian, has rank two, and has zero diagonal. Thus, 
   $$
   \langle x_kx_k^*,H\rangle =0 \qquad \text{for all } k \in \{1, \dots, M\}.
   $$
   Implying that $(x_k)_{k \in \{1, \dots, M\}}$ does not do phase retrieval, independently of $M.$

   Now let us prove the remaining part of (b). Suppose now that $s \geq 2$. Each $s$-sparse vector $x_k$ has at most $s(s-1)$ ordered pairs of distinct indices $(i,j)$ such that both $i$ and $j$ belong to $\supp(x_k)$. Hence, the $M$ vectors $x_k$ together cover at most 
   $$
   Ms(s-1)
   $$
   of the $n(n-1)$ ordered pairs $(i,j)$ with $i \neq j.$

   If
   $$
   M< \frac{n(n-1)}{s(s-1)},
   $$
   then there must exist distinct $i,j \in [1, \dots, n\}$ such that, for every $k \in \{1, \dots, M\}$, at least one of $i$ and $j$ does not belong to $\supp(x_k)$. Thus,

   $$
    (x_k)_i\overline{(x_k)_j}=0
    \qquad\text{for every }k\in\{1, \dots, M\}.
   $$
 We let
   $$
   H=e_ie_j^* +e_je_i^*,
   $$
which again yields in a nonzero Hermitian matrix of rank two, satisfying
$$
    \langle x_kx_k^*,H\rangle=0
    \qquad\text{for all }k\in\{1, \dots, M\}.
$$

By \autoref{lem:phase-retrieval-characterization}, the vectors do not do phase retrieval, which proves the claim.
\end{proof}

 Before moving on to the proof, we discuss a few remarks concerning the interpretation of \autoref{th:main}.  First, it should be stressed that the theorem is an \emph{injectivity} result. It states that under certain condition, there will exist two codebooks $Q$ and $Q'$ that cannot be distinguished through an observation of $\Phi$. It hence does \emph{not} say that there for \emph{every} given codebook $Q$ exists a $Q'$ that it is confounded with. This means that certain classes of codebooks $Q$ may be retrievable from the measurements $\Phi([Q])$ with fewer than quadratically many $x_k$. On the other hand, the theorem \emph{only} concerns injectivity. It does not say that recovery of $Q$ through an efficient algorithm is possible as soon as the number of measurements exceeds the threshold. We study these questions empirically in \Autoref{sec:numerics}. As we will see, recovery of $Q$ from the observation of $\Phi([Q])$ indeed is possible, and the quadratic threshold  seems to only be realised for special distributions of the $Q$.

 The result in part (a) for $1$-sparse plaintexts deserves a special comment. While at first glance, it might seem very strong, it should be interpreted with caution in terms of safety for the communication scheme. In the case of $1$-sparse $x_k$, the determination of $Q$ up to a global scalar is namely not necessary in order for Eve to decrypt messages.
Instead, it suffices to be able to calculate $\Phi([Q])$ for any set of $1$-sparse vectors from observations of the ciphertexts. Doing so means that every column of $Q$ is known up to a scalar. When phased with decoding a new unknown sparse message $x$, one can simply compare the new $\Phi([Q])$ with the known columns, and from that determine for which $k \in \{1, \dots, n\}$ one has $x=e_k$ modulo unit scalars.

 One should also note that we thus far only have been conserved with the idealised infinite sampling setting. We will return to the finite sampling setting once we have presented the proof of  \autoref{th:PRvsPhi}.

\section{Proof of \autoref{th:PRvsPhi}} \label{sec:proof}
In this section we prove \autoref{th:PRvsPhi}. We begin with the following auxiliary result, which allows us to circumvent the slightly inconvenient structure of the space $({\C^m/\mathbb S^1})^M$.

\begin{lem} \label{lem:nu}
    The map
    \begin{align*}
        \nu : \C^m/\mathbb S^1 &\to \C^{m,m},\\ [\lambda] &\mapsto (\overline{\lambda_i}\lambda_j)_{i,j \in \{1, \dots, m\}}
    \end{align*}
    is injective.
\end{lem}

\begin{proof}
     First, it is clear that the map $\nu$ is well defined. Namely, if $\lambda = \theta \kappa$ for some $\theta \in \mathbb{S}^1$, then 
     $$
     \overline{\lambda_i}\lambda_j = \overline{\theta\kappa_i} \theta \kappa_j = |\theta|^2\overline{\kappa_i}\kappa_j
     = \overline{\kappa_i}\kappa_j.
     $$
    
    Now suppose that $\nu([\lambda]) = \nu([\kappa])$. Inspecting the diagonals of that matrix equation, we obtain 
    $$
    |\lambda_i|^2=|\kappa_i|^2 \qquad \text{for all } i \in \{1, \dots, m\}.
    $$
    In particular, if $\lambda=0$, then also  $\kappa=0$, and thus $[\lambda]=[\kappa]$. Else, choose an index $i$ such that $\lambda_i\neq 0$. Since $|\lambda_i|=|\kappa_i|$, there exists a $\theta \in \mathbb{S}^1$ such that
    $\lambda_0 = \theta \kappa_0$. Now, since $\nu([\lambda])_{ij}=\nu([\kappa])_{ij}$, we may argue that
    \begin{align*}
\frac{\overline{\lambda_i}\lambda_j}{\abs{\lambda_i}^2} =\frac{\overline{\kappa_i}\kappa_j}{\abs{\kappa_i}^2} 
    \end{align*}
and thus
$$
 \frac{\lambda_j}{\lambda_i} = \frac{\kappa_j}{\kappa_i} \qquad \text{for all } j.
$$
    
Multiplying this equation by $\lambda_i=\theta\kappa_i$ yields $\lambda_j = \theta \kappa_j$ for every $j$. Hence, $\nu(\lambda)=\nu(\kappa)$ implies that $\lambda=\theta \kappa$, and thus $[\lambda]=[\kappa]$. Therefore, $\nu$ is injective, which is what we wanted to prove.
\end{proof}

An immediate consequence of \autoref{lem:nu} is that $\Phi$ is injective if and only if the map 
\begin{align*}
    \Psi:\C^{m\times n}/\mathbb S^1 &\to (\C^{m,m})^M, \\
    [Q] &\mapsto (\nu(Qx_k))_{k\in \{1, \dots, M\}}
\end{align*}
is injective. More explicitly,
\begin{align}
    \Psi([Q])     =
    \left(\overline{(Qx_k)_i}(Qx_k)_j
    \right)_{\substack{i,j\in\{1, \dots, m\}\\k\in\{1, \dots, M\}}}.
    \label{eq:psi}
\end{align}
With this fact in our hands, proving one direction of \autoref{th:PRvsPhi} is straightforward.

\begin{prop}\label{prop:maina1}
    If $\Phi$ is injective, then $(x_k)_{k \in \{1, \dots, M\}}$ does phase retrieval. 
\end{prop}

\begin{proof}
    We prove the contrapositive. Assume that $(x_k)_{k \in \{1, \dots, M\}}$ does not do phase retrieval. Then there exist two vectors 
    $q,w \in \C^m$ with $[q]\neq [w]$ such that $$
    \abs{\langle q, x_k\rangle} = \abs{\langle w, x_k \rangle} \qquad \text{for all } k \in \{1, \dots, M\}.
    $$
    But with $$Q=e_0q^* \qquad \text{and} \qquad W=e_0w^*,$$ where $e_0\in \C^m$ is the first standard basis vector, we get
    \begin{align*}
        \Psi([Q]) &= (\abs{\langle q,x_k \rangle}^2\delta_{0i}\delta_{0j})_{i,j\in \{1, \dots, m\}, k \in \{1, \dots, M\}} \\
        &=  (\abs{\langle w,x_k \rangle }^2\delta_{0i}\delta_{0j})_{i,j\in \{1, \dots, m\},k \in \{1, \dots, M\}} = \Psi([W]).   \end{align*}
        It is also clear that $[Q] \neq [W]$: if $Q=\theta W$ for some $\theta \in \mathbb{S}^1$, then $q=\overline{\theta}w$, which contradicts $[q] \neq [w]$. Hence, $\Psi$ is not injective, and by \autoref{lem:nu}, neither is $\Phi$, which proves the claim.
\end{proof}

To prove the other direction of \autoref{th:PRvsPhi} is considerably more involved. The first step is to realize that the recovery of $Q$ from $\Psi([Q])$ can be interpreted as a \emph{generalized phase retrieval problem}, as defined in \cite{wang2019generalized}. 

\begin{defi}[Generalized phase retrieval]
    Let $\alpha_{\ell \in \{1, \dots, L\}}$ be a set of Hermitian sesquilinear forms on a complex Hilbert space $\calH$. The problem of recovering $w\in \calH$ up to a global phase scalar from the measurements 
    $$
    \alpha_\ell(w,w),\qquad \ell \in \{1, \dots, L\}
    $$
    is a \emph{generalized phase retrieval problem}. We say that the forms \emph{do generalized phase retrieval} if the map 
    \begin{align*}
    \varphi: \calH / \mathbb{S}^1 &\to \R^L \\
    [w] &\mapsto (\alpha_\ell(w,w))_{\ell \in \{1, \dots, L\}}
    \end{align*}
    is injective.
\end{defi}

$\Psi$ does not directly define a generalized phase retrieval problem as defined above. To relate this definition to $\Psi$, consider $\C^{m\times n}$ as a complex Hilbert space and define, for $i,j\in\{1, \dots, m\}$ and $k\in\{1, \dots, M\}$, so $\Psi$ define a number of sesquilinear forms:

\begin{align*}
    \alpha_{ij}^k(Q,W) = \langle(e_i, Qx_k\rangle \langle(Wx_k,e_j\rangle.
\end{align*}

These forms are not Hermitian in general, but instead satisfy 

$$
\alpha_{ij}^k(Q,W) = \overline{\alpha_{ji}^k(W,Q)}.
$$

Le therefore define the following construction:
\begin{align*}
    \mathrm{r}\alpha_{ij}^{k} = \tfrac{1}{2}(\alpha_{ij}^k + \alpha_{ji}^k) \ \ialpha_{ij}^{k} = \tfrac{1}{2\iota}(\alpha_{ij}^k - \alpha_{ji}^k)
\end{align*}
Both forms $ \mathrm{r}\alpha_{ij}^{k}$ and $\ialpha_{ij}^{k}$ are Hermitian, and we have 
$$
\alpha_{ij}^k=\ralpha_{ij}^k+\iota \ialpha_{ij}^k.
$$

In fact, $\ralpha_{ij}^k(Q,Q)$ and $\ialpha_{ij}^k(Q,Q)$ are the real and imaginary parts of $\alpha_{ij}^k(Q,Q)$, respectively. So we have
\begin{equation*}
   \begin{matrix} \alpha_{ij}^k(Q,Q)=\alpha_{ij}^k(Q',Q')  \\
   \text{ for all $i$,$j$,$k$} \end{matrix}  \Leftrightarrow \begin{matrix} \ralpha_{ij}^k(Q,Q)=\ralpha_{ij}^k(Q',Q')  \\ \ialpha_{ij}^k(Q,Q)=\ialpha_{ij}^k(Q',Q')  \\
   \text{ for all $i, j, k$} \end{matrix}.
\end{equation*}
Thus, $\Psi$ is injective if and only if the collection of forms $\ralpha_{ij}^k,\ialpha_{ij}^k$ does generalized phase retrieval. 

In order to show the other direction of \autoref{th:PRvsPhi}, we hence need to argue that the $\alpha_{ij}^k$ do generalized phase retrieval when the $x_k$ do classical phase retrieval. To do so, we will need some fundamental results from the literature about generalized and classical phase retrieval, which we now present. We have slightly reformulated the statements to better match our notation.

\begin{theo} \label{th:genPR}
\begin{enumerate}[(a)]
    \item \cite[Theorem 2.2]{wang2019generalized}
    The forms $\ralpha_{ij}^k, \ialpha_{ij}^k$ do generalized phase retrieval if and only if there do not exist $Q,W \in \C^{m \times n}$ with 
    $$
    Q\notin \spann_\R (\iota W)
    $$ 
    such that
    $$
    \Re\bigl(\ralpha_{ij}^k(Q,W)\bigr) = \Re\bigl(\ialpha_{ij}^k(Q,W)\bigr)=0
    $$
    for all $i, j \in \{1, \dots, m\}$ and $k \in \{1, \dots, M\}.$
    
    \item (\cite[Theorem 4]{bandeira2014saving})
    If $(x_k)_{k \in \{1, \dots, M\}}$ does phase retrieval, then
    \begin{align*}
        \spann_\R\{x_kx_k^*u:k\in\{1, \dots, M\}\} = \spann_\R(\iota u)^\perp
    \end{align*}
    for every $u\in \C^n$. That is, the \emph{real} span of the vectors $(x_kx_k^*u)_{k \in\{1, \dots, M\}}$ is equal to the orthogonal complement with respect to the real inner product $\Re(\langle \cdot,\cdot\rangle)$) of the vector $\iota u$.

    \item (\cite[Theorem 7]{bandeira2014saving},\cite[Theorem 2.2]{grohs2020phase}) If $(x_k)_{k \in \{1, \dots, M\}}$ does phase retrieval, then it has the \emph{complement property}: for every subset $S$ of $\{1, \dots, M\}$, either $(x_k)_{k \in S}$ or $(x_k)_{k \in S^c}$ spans $\C^n$. 
\end{enumerate}
\end{theo}

We now have all the tools to show the other direction of \autoref{th:PRvsPhi}.

\begin{prop}\label{prop:maina2}
    If $(x_k)_{k \in \{1, \dots, M\}}$ do phase retrieval, then  $\Phi$ is injective.
\end{prop}

\begin{proof} 
    By \autoref{th:genPR}(a), \autoref{lem:nu}, and the discussion above, it is enough to prove that for any pair of matrices $Q,W \in \C^{m \times n}$ satisfy
    $$
    \Re(\ralpha_{ij}^{k}(Q,W))=\Re(\ialpha_{ij}^{k}(Q,W))=0,
    $$ we must have $W\in \spann_\R(iQ)$. If $Q=0$, it is surely in $\mathrm{span}_\R(\iota W)$, so that case can be excluded.

    Write $$
    q_i = Q^*e_i, \qquad  w_i = W^*e_i, \qquad i \in \{1, \dots, m\}.
    $$
    We begin by assuming that all $q_i\neq 0$. We will return to the other case later.

    Taking $i=j$, the condition $\Re(\ralpha_{ii}^{k}(Q,W))=0$ gives
    \begin{align*}
        0 = \Re\bigl(\sprod{q_i,x_k}\sprod{x_k,w_i}\bigr) = \Re\bigl(\sprod{x_kx_k^*q_i,w_i}\bigr)
    \end{align*}
    for every $k \in \{1, \dots, M\}.$
    
    By Theorem \ref{th:genPR}(b), $q_i\neq 0$, and the fact that the $x_k$ do phase retrieval
   $$
   \spann_\R\{x_kx_k^*q_i:k\in\{1, \dots, M\}\} =
    \spann_\R(\iota q_i)^\perp.
   $$
   Hence, by taking real linear combination of the above equation, we obtain that
   $$
   \Re\sprod{v,w_i}=0 \qquad
    \text{for all }
    v\in\spann_\R(\iota q_i)^\perp,
   $$
   and thus
$$
w_i \in \spann_\R(\iota q_i)
$$

Hence, for every $i\in\{1, \dots, m\}$, there exists $\lambda_i\in\R$ such that
$$w_i=\lambda_i\iota q_i.
$$

It remains to show that the $\lambda_i$ can be chosen independently of $i$. 
Substituting $w_i=\lambda_i\iota q_i$ and $w_j= \lambda_j \iota q_j$ into 
$$
\Re\bigl(\ralpha_{ij}^k(Q,W)\bigr)=0, \qquad
\Re\bigl(\ialpha_{ij}^k(Q,W)\bigr)=0,
$$
and defining 
$$
a_{ij}^k := \sprod{q_i,x_k}\sprod{x_k,q_j},
$$
gives
    \begin{align}
    \lambda_j a_{ij}^k + \lambda_i \overline{a_{ij}^k} &\in \R, \label{eq:aqre}\\
    \lambda_j a_{ij}^k
    - \lambda_i \overline{a_{ij}^k}
    &\in \iota\R. \label{eq:aqim}
    \end{align}
    
    By simply looking at the real and imaginary parts separately and utilizing that $\lambda_i,\lambda_j \in \R$, we see from \eqref{eq:aqre} that either $\lambda_i = \lambda_j$ or $a_{ij}^k\in \R$, and that $\eqref{eq:aqim}$ implies that either $\lambda_i = \lambda_j$ or $a_{ij}^k\in \iota \R$,     \eqref{eq:aqim} implies that either $\lambda_i=\lambda_j$ or $a_{ij}^k\in\iota\R$.
    
    Thus, if $\lambda_i\neq \lambda_j$, then
    $$
    a_{ij}^k \in \R \cap \iota\R= \{0\} \qquad \text{for every } k \in \{1, \dots, M\}.
    $$
    But that means 
    $$
    \sprod{q_i,x_k}\sprod{x_k,q_j}=0 \qquad \text{for all } k \in \{1, \dots, M\}.
    $$

Let 
$$
        S:=\{k\in\{1, \dots, M\}:\sprod{q_i,x_k}=0\} \sse \{1, \dots, M\}.
$$
For every $k\in S^c$, we then have $\sprod{x_k,q_j}=0$. By the complement property from \autoref{th:genPR}(c), either
$(x_k)_{k \in S}$ or $(x_k)_{k\in S^c}$ spans $\C^n$. In the first case, $q_i=0$, while in the second case $q_j=0$, both of which we have ruled out. Since $i$ and $j$ are arbitrary, we have shown that 
$$\lambda_i =\lambda_j$$ whenever $q_i, q_j \neq 0.$

It follows that there exists a single $\lambda\in\R$ such that
    $$
        w_i=\lambda\iota q_i
    $$
for every $i$ with $q_i\neq 0$. 

It remains the case when certain $q_i=0$. In that case, we however have $\alpha_{ij}^k(Q,W)=0$ for all $j$ and $k$. The equations 
\begin{align*}
\Re\bigl(\ralpha_{ij}^{k}(Q,W)\bigr)&=0, \\
\Re\bigl(\ialpha_{ij}^k(Q,W)\bigr)&=0
\end{align*}
therefore imply that 
$$
\alpha_{ji}^k(Q,W)=0
$$
for all $j$ and $k$, i.e.,
    \begin{align*}
        \sprod{q_j,x_k}\sprod{x_k,w_i}=0 \; \text{for all }j\in\{1, \dots, m\},\ k\in\{1, \dots, M\}.
    \end{align*}

Since by assumption $Q\neq 0$, there exists some $j$ such that $q_j \neq 0.$ Applying the complement property as above to this $q_j$ and $w_i$ shows that either $q_j=0$ or $w_i=0$. Since $q_j \neq 0$ by assumption, we must have $w_i=0$. Thus, for every $i$ with $q_i=0$,
$$
w_i= \lambda \iota q_i
$$
for any given $\lambda$. 
Together with the previous case, this gives
$$
    w_i=\lambda\iota q_i
    \qquad\text{for all }i\in\{1, \dots, m\}.
$$
for some common $\lambda \in \R$, and consequently,
$$
    W=-\lambda\iota Q\in\spann_\R(\iota Q),
$$
which concludes our proof.
\end{proof}

This concludes the proof of \autoref{th:PRvsPhi}, since it is the combination of \autoref{prop:maina1} and \autoref{prop:maina2}.

\section{Learning \texorpdfstring{$\Phi([Q])$}{Phi(Q)} from finite samples} \label{sec:finite_samples}
Until now, Eve has been granted an idealized advantage: unlimited access to the distributions $y_k=h*Qx_k$. This assumption is analytically convenient but practically unrealistic, as it requires arbitrarily many encryptions of each plaintext. In this section, we drop this idealization and ask if Eve's attack survives when instead we only let her observe finitely many samples of each $y_k$.

\subsection{A no-go case} 
We first show that if Bob and Alice are careless in their design of $Q$ and $\mathcal{D}$, namely if $h$ has its nonzero values concentrated, the filter provides no additional security whatsoever. 

The idea is that, for certain configurations of $Q$ and $h$,  Eve can use the blind deconvolution algorithm in \cite{pulse_stream_model} directly to recover $h*Qx_k$ modulo unit scalar, without ever drawing samples of $h$ -- making the role of the filter irrelevant. The algorithm in \cite{pulse_stream_model} exploits a simple separation property and is designed to retrieve a pair $(u,v)$ from only the observation of $u*v$. For it to work, we need to make strong structural assumptions on $u$ and $v$, namely that the non-zero entries in $u$ all are concentrated in a short block of indices, while the nonzero indices in $v$ are sufficiently spread out. Then, the convolution $u*v$ will consist a small number of separated copies of $u$, each multiplied with a different entry $v(i)$, written next to each other. From this, $u$ and $v$ can be read directly up to the unavoidable scalar ambiguity. See also Figure \ref{fig:separation}

\begin{rem} The authors of \cite{pulse_stream_model} go further in that they also consider the case where $u*v$ is only partially observed. Here, we only use the fully observed case as a theoretical tool.
\end{rem}

Let us formalize the aforementioned structural assumptions.

\begin{defi}\label{def:vec_Q_nogov2}
    \begin{itemize}
        \item [(a)] A vector $v \in \mathbb{R}^n$ is \emph{$(\Delta,S)$-separated} if $\abs{\supp(v)}\le S$, and writing
        $$
        \mathrm{supp}(v)=\{i_1, \ldots, i_s\}, \; i_1 <i_2<\cdots <i_s
         $$ 
        satisfies 
        $$
        i_{k+1} - i_k > \Delta \qquad \text{for all } i_k \in \{1, \cdots, s\}.
        $$ 
        We denote the set of all such vectors by $\mathcal{M}(\Delta,S)$.
        \item [(b)] A pair $(u,v)$ is $(\sigma,S,\Delta)$-structurally sparse if 
        $$
    \supp(u)\subseteq\{1,\ldots,\sigma\}, \qquad \text{where } \sigma\le\Delta
$$ and
$$v\in\mathcal M(\Delta,S),$$
 i.e., $u$ has its nonzero entries concentrated among its first $\sigma$ indices, where $\sigma \leq \Delta$ and $v$ is $(\Delta,S)$-separated.
    \end{itemize}
    \end{defi}

The algorithm in \cite{pulse_stream_model} applies to a pair $(h,Qx_k)$ whenever it is $(\sigma,S,\Delta)$-structurally sparse. The following condition on the design of $Q$ guarantees that this will be the case for all values of $x_k$.

\begin{defi}
For a matrix $Q \in \R^{m \times n}$, let $q_i$ denote the $i$-th row of $Q$, and let
$$
I_Q:=\{i\in\{1,\ldots,m\}: q_i\neq 0\}
$$
denote the set of indices of its nonzero rows.

We say that $Q$ satisfies the $(\Delta,S)$-structure if $\abs{I_Q} \leq S$ and writing 
$$
 I_Q=\{i_1,\ldots,i_s\}, \qquad i_1<\cdots<i_s,
$$
we have
$$
i_{k+1}-i_k>\Delta
\qquad\text{for all }k\in\{1,\ldots,s\}.
$$
\end{defi}

Indeed, if we assume that $Q$ satisfies the $(\Delta,S)$-structure, and $h$ has its nonzero entries concentrated among the first $h$-zero indices, then $(h,Qx_k)$ is $(\sigma,S,\Delta)$-structured, and Eve can determine $Qx_k$ modulo a scalar from a single observation $h*Qx_k$. A formal statement is as follows.

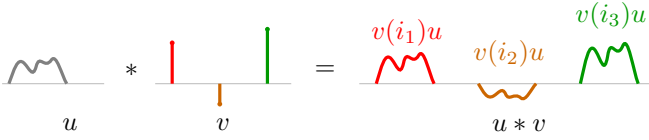
\begin{figure}
    \centering
    \begin{tikzpicture}[scale=0.45] 
\node at (2,-1.2) {$u$}; \draw[gray!50] (0,0) -- (3,0); \draw[very thick,gray] (0.2,0) .. controls (0.35,0.4) and (0.45,0.8) .. (0.65,0.6) .. controls (0.8,0.45) and (0.9,0.15) .. (1.0,0.5) .. controls (1.1,0.9) and (1.25,0.4) .. (1.45,0.7) .. controls (1.6,0.9) and (1.75,0.25) .. (1.9,0); \node at (3.8,0.4) {$*$}; 
\node at (6.5,-1.2) {$v$}; \draw[gray!50] (4.5,0) -- (8.5,0); \draw[red,very thick] (5.0,0) -- (5.0,1.2); \fill[red] (5.0,1.2) circle (2pt); \draw[orange!80!black,very thick] (6.4,0) -- (6.4,-0.6); \fill[orange!80!black] (6.4,-0.6) circle (2pt); \draw[green!60!black,very thick] (7.8,0) -- (7.8,1.6); \fill[green!60!black] (7.8,1.6) circle (2pt);
\node at (15.2,-1.2) {$u*v$}; \draw[gray!50] (10.5,0) -- (20.5,0); 
\draw[very thick,red] (11.0,0) .. controls (11.15,0.48) and (11.25,0.96) .. (11.45,0.72) .. controls (11.60,0.54) and (11.70,0.18) .. (11.80,0.60) .. controls (11.90,1.08) and (12.05,0.48) .. (12.25,0.84) .. controls (12.40,1.08) and (12.55,0.30) .. (12.70,0); 
\draw[very thick,orange!80!black] (14.0,0) .. controls (14.15,-0.24) and (14.25,-0.48) .. (14.45,-0.36) .. controls (14.60,-0.27) and (14.70,-0.09) .. (14.80,-0.30) .. controls (14.90,-0.54) and (15.05,-0.24) .. (15.25,-0.42) .. controls (15.40,-0.54) and (15.55,-0.15) .. (15.70,0);
\draw[very thick,green!60!black] (17.0,0) .. controls (17.15,0.64) and (17.25,1.28) .. (17.45,0.96) .. controls (17.60,0.72) and (17.70,0.24) .. (17.80,0.80) .. controls (17.90,1.44) and (18.05,0.64) .. (18.25,1.12) .. controls (18.40,1.44) and (18.55,0.40) .. (18.70,0);
\node[red] at (11.9,1.5) {$v(i_1)u$}; \node[orange!80!black] at (14.9,0.9) {$v(i_2)u$}; \node[green!60!black] at (17.9,2.0) {$v(i_3)u$}; 

\node at (9.5,0.4) {$=$};
\end{tikzpicture}
    \caption{Illustration of the separation property used in \autoref{theo:no_go_case}. Note that if the support of $v$ is not sufficiently separated, the copies of $u$ in $u*v$ will start overlapping each other, making the algorithm not applicable.}
    \label{fig:separation}
\end{figure}

\begin{theo}\label{theo:no_go_case}
Suppose $Q$ satisfies the $(\Delta,S)$-structure, and that 
$$
\supp(h)=\{1,\ldots,\sigma\}
    \qquad\text{for some }\sigma\leq\Delta,
$$
meaning $h$ has all nonzero values among its first $\sigma \leq \Delta$ entries. Further suppose that the norm of $h$ is a known positive constant.
Then Eve learns the mapping $\Phi([Q])$ in \autoref{lem:PhiMap} from a single observation of each 
$$
y_k = h*(Qx_k), \qquad k\in\{1,\ldots,M\}.
$$
\end{theo}

\begin{proof}
Let us fix $k \in \{1, \dots,M\}.$ Since $Q$ satisfies the $(\Delta,S)$-structure, 
$$
\supp(Qx_k) \subseteq I_Q.
$$
Indeed, if $q_i=0,$ then $(Qx_k)_i=q_i x_k=0.$
Thus, $Qx_k$ has at most $S$ nonzero entries, whose indices are separated by more than $\Delta$, and thus,
$$Qx_k \in \mathcal{M}(\Delta,S),$$ 
see also \cite[Lemma 2]{matilda_thesis}.

If $Qx_k=0$, then $y_k=0$ and Eve immediately recovers $Qx_k.$
Otherwise, write
$$
\supp(Qx_k) = \{i_1, \dots, i_t\}, \qquad  i_1 < \cdots < i_t,
$$
where $ t \leq S$ and
$$
i_{j+1}-i_j > \Delta \qquad\text{for all }j\in\{1,\ldots,t-1\}.
$$
By assumption 
$$
\supp(h)=\{1,\ldots,\sigma\}, \qquad \sigma \leq \Delta.
$$
Hence, $(h,Qx_k)$ is $(\sigma,S,\Delta)$-structurally sparse.
In particular, the locations of the blocks determine $\supp(Qx_k)$, while their relative scalings determine the ratios between the corresponding nonzero entries of $Qx_k$. Indeed, since the nonzero entries of $h$ are concentrated in its first $\sigma$ indices, convolving with $Qx_k$ spreads each nonzero entry $Qx_k(i_j)$ into a block of length $\sigma$. Taking the block corresponding to $i_1$ as a reference, the block corresponding to $i_j$ determines
$$
    \frac{Qx_k(i_j)}{Qx_k(i_1)} \qquad\text{for all }j\in\{1,\ldots,t\}.
$$
Hence, $Qx_k$ is determined up to a common nonzero scalar.
Now, since the norm of $h$ is a known constant, this scalar is even in $\mathbb{S}^1$ and Eve therefore recovers the projective class
$$ [Qx_k].$$
Repeating the argument for every $k\in\{1,\ldots,M\}$, Eve obtains all projective measurements $[Qx_k]$. By \autoref{lem:PhiMap}, these determine $\Phi([Q])$. Hence, Eve learns $\Phi([Q])$ from a single observation of each $y_k$.

\end{proof}

We emphasize that \autoref{theo:no_go_case} does not say that recovering $\Phi([Q])$ from finitely many samples is easy in general. Instead, it only identifies specific choices of $Q$ and $\calD$ for which a single observation per plaintext suffices to recover $\Phi([Q])$, and Eve hence has a clear path to attacking their communication. While it is easy that such choices can be avoided, the result shows that the hardness of recovering the codebook from finitely many observations will depend on both the filter distribution $\calD$, as well as the structure of the measurement matrix $Q$.

\subsection{A Covariance-based Attack} \label{sec:covariance_attack}

Let us now present a simple algorithm for estimating $\Phi([Q])$ from finitely many samples, based on covariance estimation. 
The algorithm is again based on the Fourier convolution theorem, which says that convolution in the time domain corresponds to pointwise multiplication in the Fourier domain:
$$
\widehat{u*v}=\widehat{u}\odot\widehat{v}.
$$

Letting
$$
w_k=Qx_k,
$$
we have
$$
\widehat{y}_k=\widehat{h}\odot \widehat{w}_k=\Diag(\widehat{w}_k)\widehat{h}.
$$
If we consider $w_k=Qx_k$ as fixed, the covariance matrices $\Sigma_{\widehat y_k}$ of ${\widehat{y_k}}$ and $\Sigma_{\widehat h}$ of $\widehat{h}$, respectively, satisfy
\begin{equation*}
    \Sigma_{\widehat y_k}= \Diag(\widehat{w}_k) \Sigma_{\widehat h}{\Diag(\widehat{w}_k)}^*.
\end{equation*}

Equivalently,
$$
\mathcal{F}(\Sigma_{y_k})=
    \Diag(\mathcal FQx_k)
    \bigl(\mathcal F(\Sigma_h)\bigr)
    \Diag(\mathcal FQx_k)^*.
$$
Using our convention for the Fourier transform $\calF(A)\widehat{x}=\widehat{Ax}$, this can also be written as
$$
\mathcal{F}(\Sigma_{y_k})= \calF(\Sigma_h)\odot\calF(w_kw_k^*).$$

Hence, letting $\oslash$ denote the element-wise division, whenever $\mathcal{F}(\Sigma_h)$ has no vanishing entries, we obtain
$$
(\mathcal{F}\Sigma_y) \oslash (\mathcal{F}\Sigma_h) = \mathcal{F}(w_k w_k^*).
$$
Taking the inverse Fourier transform then isolates the rank-one matrix $w_k w_k^*$:
\begin{equation*}
\mathcal{F}^{-1}
\left(
(\mathcal{F}\Sigma_y) \oslash (\mathcal{F}\Sigma_h)
\right)= w_kw_k^*.
\end{equation*}
By extracting the leading eigenvector of this matrix, $w_k$ is then recovered up to a global phase. 

Now, the covariance $\Sigma_h$ is known from the public filter distribution $\calD$, and $\Sigma_{y_k}$ is simple to estimate from repeated encryptions of the same plaintext $x_k$. Hence, as long as $\mathcal{F}\Sigma_h$ has no vanishing entries, Eve will be able to estimate $Qx_k $ modulo unit scalars from finitely many observations of $y_k$. The Pseudocode for the entire algorithm is given in \Autoref{algo}.

\begin{algorithm}
\caption{Covariance-based recovery of $Qx_k$}
\label{algo}
\begin{algorithmic}[1]
\Require $r$ ciphertext samples $y_{k_1}, \ldots, y_{k_r}$ with $y_{k,p}=h_{k,p}*Qx_k$, and known covariance $\Sigma_h$ 
\Ensure Estimate of $Qx_k$ up to a global phase

\State Compute the sample covariance matrix of $y$\footnote{Note that we here can assume that $h$, and therefore $y$, has zero mean, due to the symmetry assumption}:
\[
    \overline{\Sigma}_{y_k} = \frac{1}{r} \sum_{p=1}^r y_{k,p}y_{k,p}^*
\]

\State Transform the covariance matrices to the Fourier domain:
$$
\overline{C}_{y_k} = \mathcal F \overline{\Sigma}_{y_k} , \qquad C_h = \mathcal F \Sigma_h.
$$

\State Compute the element-wise ratio in the Fourier domain:
\[
    M = (\mathcal{F}\overline{\Sigma}_y)\oslash (\mathcal{F}\Sigma_h)
\]
\State Recover the rank-1 matrix via the inverse Fourier transform:
\[
    R_k = \mathcal{F}^{-1}M \approx (Qx_k) (Qx_k)^*
\]
\State Compute the SVD of $R_k$, extract the leading singular vector and let $(\overline{\lambda}_1,\overline{u}_1)$ be its leading eigenpair.
\State \textbf{Return}
$$
    \overline{Qx_k}
    =
    \sqrt{\overline{\lambda}_1}\,\overline{u}_1.
$$
\end{algorithmic}
\end{algorithm}

It is well known that $\sim \epsilon^{-2}$ samples are needed to obtain an estimate of a covariance with error less than $\epsilon$. Concretely, we have the following result.

\begin{theo}\label{theo:sample_complex_covy} 
Let 
$$
K_{\Sigma} = c_4 \frac{8}{3} \Vert w_k\Vert_2^2 \Vert \Sigma_h \Vert,
$$
where $c_4>0$ is a constant. For $\epsilon < K_{\Sigma}$ and $\delta \in (0,1)$, if the sample size is at least
$$
r \geq \frac{K_{\Sigma}^2}{\epsilon^2}\ln\left(\frac{2m^2}{\delta}\right).
$$
then any element of $\overline{\Sigma}_{y}$ is within $\epsilon$ of its true value with probability $(1-\delta)$.
\end{theo}

The proof of \autoref{theo:sample_complex_covy} is derived in detail in \cite{matilda_thesis}. Its proof relies on the subexponential Bernstein inequality (see for example \cite{Bernstein_ineq}), applicable since each entry of $\overline{\Sigma}_y$ is a sum of subexponential random variables.

This algorithm is simple, but its existence shows that as long as $\mathcal{F}\Sigma_h$ has no vanishing entry, Eve can in principle obtain $[Qx_k]$ from finitely many ciphertext samples. Repeating the procedure for the plaintexts $(x_k)_{k \in [M]}$ gives the measurements required to recover $\Phi([Q])$ as in \autoref{lem:PhiMap}. The practicality of the attack boils down to how many samples she needs, and this boils down to how high fidelity the estimate of $\Phi([Q])$ needs to be. Note that the theorem predicts that the recovery may be quite heavy in $r$ -- to reach an error of $\epsilon$, $\sim \epsilon^{-2}$ samples are needed.  We investigate this empirically in \Autoref{sec:numerics}.

Equally important to note is however that the algorithm is easy to defend against: First, the bound in \autoref{theo:sample_complex_covy} immediately shows that by increasing $\|\Sigma_h\|$, more samples are need to obtain the same error. Secondly, it fails completely when $\mathcal{F}\Sigma_h$ contains zero values. It turns out this is the case  for several natural choices of the filter distribution, as we show next.

\begin{theo}\label{theo:main_sparse}
Consider the model $h=z\odot b$, as in \autoref{eq:GaussianFill}. Suppose that either
\begin{enumerate}[(a)]
    \item the entries of $b$ are i.i.d. centered, and the entries of $z$ have equal second moments,
    \item the entries of $z$ are i.i.d. centered,  and the entries of $b$ have equal second moments. 
\end{enumerate}  
Then $\Sigma_h$ is a scalar multiple of the identity. Therefore $\mathcal{F}\Sigma_h$ is also a scalar multiple of the identity, and thus, has no vanishing off-diagonal entries. Hence \Autoref{algo} cannot recover $Qx_k$ using the covariance-based method.
\end{theo}

\begin{proof}
By the law of total covariance \cite{law_of_totalcov},
\begin{equation}\label{eq:Cov_h_basic}
\Cov(h) = \mathbb{E}[\Cov(h \vert z)] + \Cov(\mathbb{E}[h \vert z]).
\end{equation}

From \eqref{eq:GaussianFill}, 
$$
\Cov(h \vert z) = \Diag(z) \Sigma_b \Diag(z).
$$
Taking expectations pointwise gives 
$$
\mathbb{E}[\Cov(h \vert z)] = \Sigma_{b} \odot \mathbb{E}[zz^T].$$
Analogously, 
$$
\Cov(\mathbb{E}[h \vert z]) = \Diag(\mu_b) \Sigma_z \Diag(\mu_b) = \mu_b \mu_b^T \odot \Sigma_z.
$$ Substituting into \eqref{eq:Cov_h_basic}, gives us

\begin{equation}\label{Eq:sigma_h}
\Sigma_h = \Sigma_{b} \odot \mathbb{E}[zz^T] + \mu_b \mu_b^T \odot \Sigma_z;
\end{equation}
in particular, $\mu_b = 0$ reduces this to $\Sigma_h = \Sigma_b \odot \mathbb{E}[zz^T]$. 

Now, assume (a). Since the entries of $b$ are centered i.i.d., $\Sigma_b$ is a scalar multiple of the identity matrix. Also, by the equal moments assumptions,  the diagonal entries of $\mathbb{E}(zz^T)$ are also equal. Equation \eqref{Eq:sigma_h} therefore shows that $\Sigma_h$ also is a scalar multiple of the identity, and consequently also $\mathcal{F}\Sigma_h$. 

The argument is similar in case (b) -- here, due to the i.i.d centered distributions of the $z_i$, $\mathbb{E}[zz^T]=\Sigma_z$ is a multiple of the identity matrix. Since $\Sigma_b$ in this case have equal diagonal entries, $\Sigma_h$ will again be a multiple of the identity matrix.
\end{proof}

\begin{rem} While \autoref{theo:main_sparse} yields a way of defending against \Autoref{algo}, one should note that  
we do not claim Eve can't find another way to estimate $Qx_k$. To analyse whether alternative finite-sample attacks are possible for these filter distributions is an interesting question, but beyond the scope of this work.
\end{rem}

\begin{rem}
    In \cite{matilda_thesis}, it is shown (specifically in Lemmas 11 and 13) that the algorithm can also be rendered inapplicable when $\mu_b \neq 0$. However, in this case, the phase symmetry assumption is violated, which is why we do not present it in detail here.
\end{rem}

\section{Numerical experiments} \label{sec:numerics}
In light of the discussion at the end of \Autoref{sec:security_analysis}, it is important to empirically test whether $Q$ can be recovered from channel observations. We do so by considering two settings in the experiments. In the first, the idealized infinite sample setting, Eve is given $\Phi([Q])$ directly. This allows us to compare practical recovery with the information-theoretic identifiability results developed in \Autoref{sec:security_analysis}. 
In the second,  the finite-sample setting,  $\Phi([Q])$ must first be estimated from ciphertext observations using the covariance-based attack described in \Autoref{sec:covariance_attack}.

In all experiments, we draw random $s$-sparse vectors $x_k$ by first drawing their support uniformly at random among all supports of size $s$, and then populating the selected coordinates with i.i.d. Gaussians, for different values of $s$. We specify the mechanism for generating $Q$ later, since it differs between experiment runs.

We are quite lenient in deeming a recovery successful: we do so already when the relative error modulo global phase satifies
$$
\min_{\abs{\theta}=1} \frac{\|Q - \theta Q_{\text{recovered}}\|}{{\|Q\|}} < 0.1.
$$
We do so to take into account the fact that, any practical scenario, involves a quantization step, meaning that a rough estimate of $Q$ should be likely to suffice for the practical resolution of the blind deconvolution afterwards. Note that this is, if anything, in \emph{Eve's favor}.

\subsection{Recovery from Infinite Samples}
In the infinite sampling regime, we are tasked with recovering $Q$ directly from the measurements 
$$
\Phi([Q])\in (\C^m/\mathbb S^1)^M.
$$
Our approach to do this is simple: we try fit the values $\Phi(Q)_k= Qx_k$ to the values $b_k$, taking the global phase ambiguity for each vector into account. That is, we minimize the loss function 
\[
\mathcal{L}(Q) = \sum_k \min_{\abs{\theta_k}=1}\left\| \Phi(Q)_k - \theta_k b_k \right\|_2^2.
\]
Note that $\mathcal{L}$ is well defined independently of the choice of the representative $b_k$, and can hence be calculated also when each $b_k$ is only known modulo a global unit scalar. 

To minimize $\calL$, we use the L-BFGS-B optimization algorithm (as implemented in numpy) \cite{liu1989limited}, a quasi-Newton method designed for high-dimensional optimization problems with variable constraints. Unlike full-memory quasi-Newton methods, which require storing and updating the entire Hessian matrix, L-BFGS-B uses a limited-memory approach, retaining only a small subset of gradient and variable updates, making it a good choice for solving large-scale optimization problems. Its effectiveness in addressing similar relative phase retrieval problems has been demonstrated in prior research, which justifies its selection for this study\cite{kornprobst2021phase,Ji2016phase}.

To make the experiments feasible, we consider relatively small parameter values: $(n=50,m=5)$ and $(n=100,m=3)$. For different values of messages $M$ and samples per message $r$ we draw $500$ instances of $Q$ and $(x_k)$, where  $Q$ is drawn as a Gaussian random matrix. We apply L-BFGS-B to minimize $\calL$, and record the error.

The results are plotted in \Autoref{fig:infinite_r}\footnote{
    The results differ slightly from the corresponding experiment in the  conference version \cite{flinth2026bilinear} of this article. The reason is an improvement of the code this time around -- the previous experiments used a derivative-free version of the L-BFGS-B algorithm, whereas we provide it symbolically this time around. This has both sped up experiments, and improved recovery slightly. }. There are two clear observations to be made: First, the recovery does get harder with smaller $s$ -- and in particular impossible for $s=1$, but the effect saturates quite early. This part of \autoref{th:main} is hence confirmed. It however also shows that it does not tell the full story. \autoref{th:main} in particular implies that $\Phi$ cannot be injective for $M<4n$, but recovery occurs \emph{much earlier} than that in both settings. As we already discussed, this does not contradict \autoref{th:main}, but it is sobering.

    \begin{figure}
    \centering
    \includegraphics[width=0.7\linewidth]{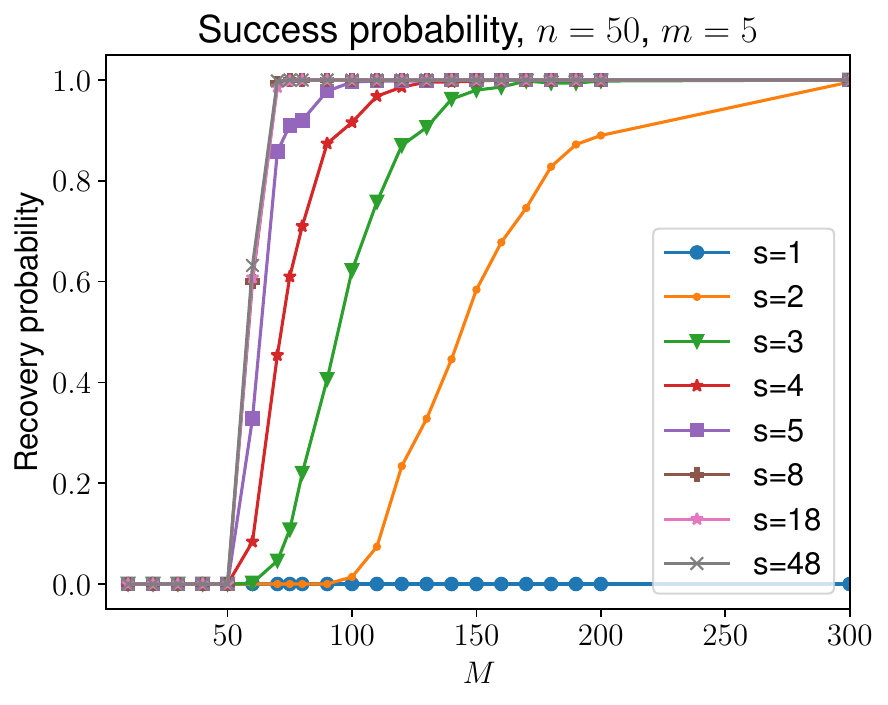} 
    
    \includegraphics[width=0.7\linewidth]{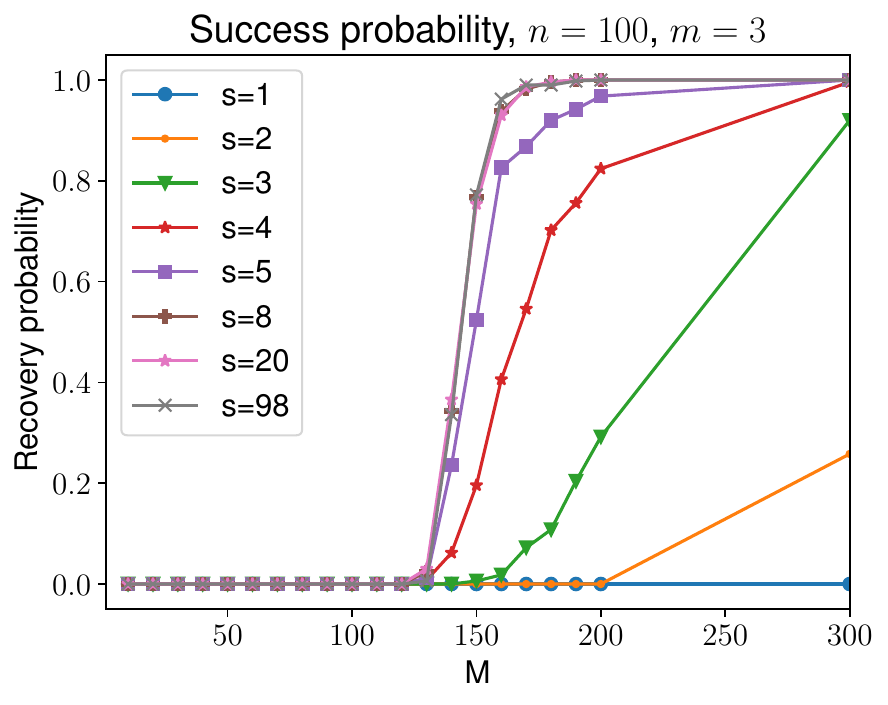}
    \caption{Success probabilities for $r=\infty$ and different $s$ and $M$. Top: $n=50$, $m=5$. Bottom: $n=100,m=3$.}
    \label{fig:infinite_r}
\end{figure}

    As we have discussed above, \autoref{th:PRvsPhi} and \autoref{th:main} do not rule out that there are particular classes of codebooks that can be recovered when  $(x_k)$ do not do phase retrieval, and that the Gaussian matrices very likely belong to these sets. Studying the proofs of these theorems, in particular that of \autoref{prop:maina1}, reveals two properties that might make recovery harder.
    \begin{enumerate}[(i)]
        \item Having only \emph{one nonzero row} -- the matrices used in the proof of \autoref{prop:maina1} have this property.
        \item Having sparse rows -- the matrices constructed in the proof of \autoref{th:main} have this property, and these matrices directly reveal vectors that are confounded by the $(x_k)$ \cite{bandeira2014saving}.
    \end{enumerate}
    This motivates us to test these two properties separately. 
    
    We thus test the recoverability of \emph{four} types of matrices: matrices with one sparse row ((i) and (ii)), many sparse rows (just (ii)), one dense row (just (i)) and many dense rows (neither (i) nor (ii)). 
    We here fix $s=3$, $m=3$ and vary $n$ between $25$ and $200$ and $M$ between $100$ and $10000$. The rows that we generate for matrices of type (ii) only have two non-zero entries, and the indices of the rows for matrices of type $(i)$ are drawn uniformly. Any non-zero entries are generated as i.i.d Gaussians.

    In \Autoref{fig:recovery_vs_n}, we plot the success rates of the experiments against the two values
    $$\frac{M}{4n} \qquad \text{and } \qquad \frac{M}{(n/s)^2},$$
    that correspond to linear and and quadratic scaling, respectively. We see that in the cases (both one and multiple) of \emph{sparse} rows, the recovery threshold fits much better to a logistic curve using (the logarithm) of $M/(n/s)^2$ as a feature than $M/(4n)$. This is also confirmed when we calculate the average distance between the measured and the (by the logistic regression) predicted probabilities as measured by KL-divergence on the data -- it is consistently lower for the fits based on the quadratic features for the sparse row codebooks, and lower for the linear feature fit for the others. Hence, it is first and foremost property (ii) that seems to be important to realize the quadratic growth with $n$. 
    
    We conclude that the quadratic scaling predicted by \autoref{th:main} does not seem to be realized for dense codebooks, but that it does emerge when the $Q$ are drawn to have sparse rows. Since our pipeline is dependent on the blind deconvolution problem being solvable, and that such guarantees almost always assume that the $Q$ is drawn according to a dense distribution, this puts an asterisk on our theoretical results. We will discuss this later, but let us first present the results for the finite sample setting.  
    
\begin{figure}
    \centering
    \fbox{\includegraphics[width=0.45\linewidth]{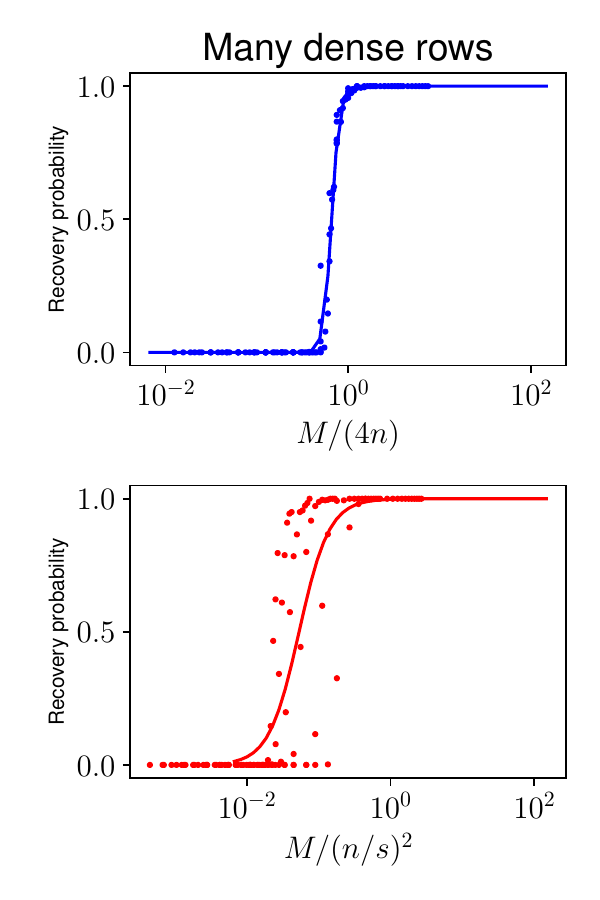}}\fbox{\includegraphics[width=.45\linewidth]{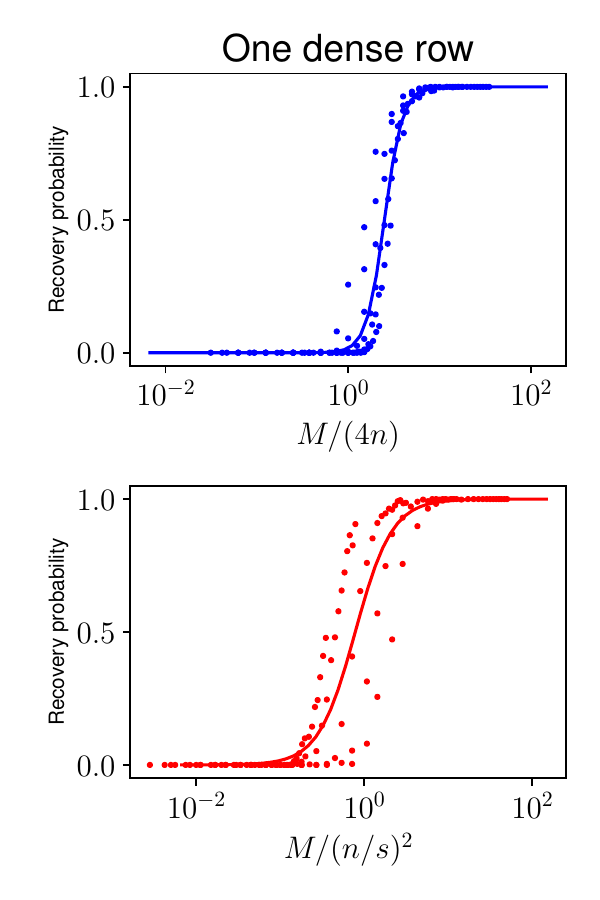}}

    \fbox{\includegraphics[width=0.45\linewidth]{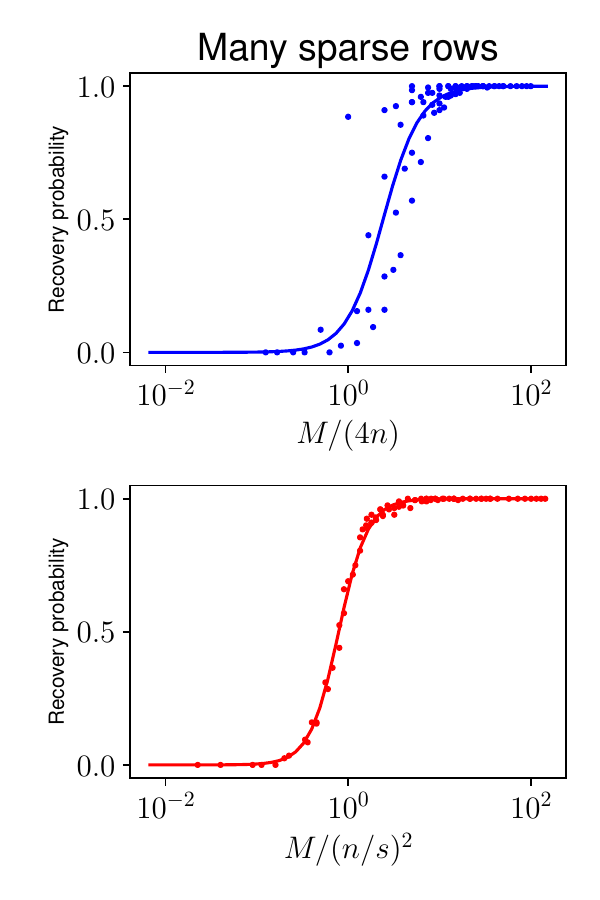}}\fbox{\includegraphics[width=.45\linewidth]{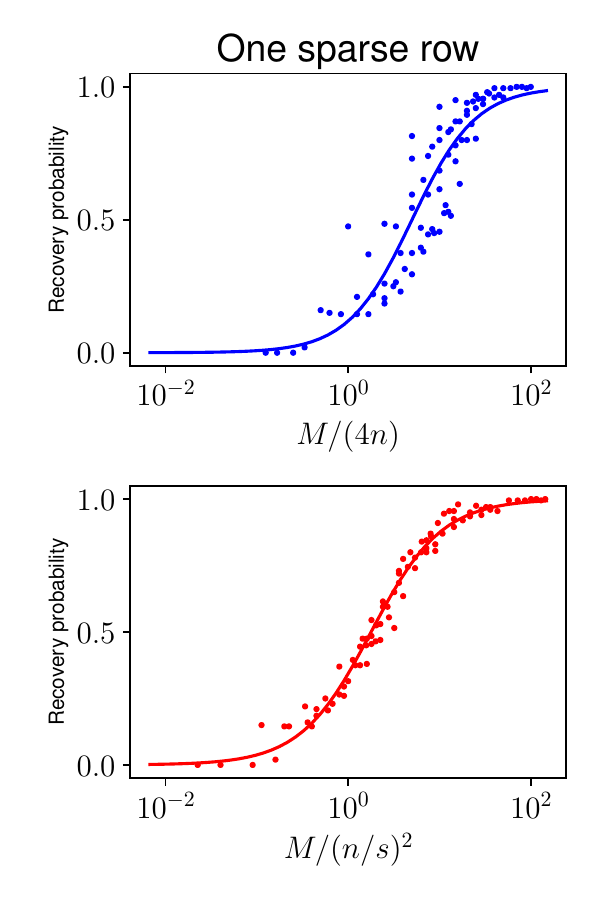}}
    \caption{Recovery for the four classes of $Q$. Every point in the scatter plots correspond to one set of experiments, and the curves are logistic regression fits. In the two cases with dense rows (above), the recovery is better explained by a linear scaling with $n$ of $M$, whereas it with sparse rows (below) is better explained by a quadratic scaling. }
    \label{fig:recovery_vs_n}
\end{figure}

\begin{table}
\centering
\begin{tabular}{l | c c}
        & Linear & Quadratic \\
\hline Many dense rows   &  $\mathbf{8.355\cdot 10^{-3}}$ & $1.706\cdot 10^{-1}$ \\
 One dense row &  $\mathbf{3.010\cdot 10^{-2}}$ & $6.604\cdot 10^{-2}$ \\
\hline Many sparse rows & $ 5.920\cdot 10^{-2}$ & $\mathbf{2.886\cdot 10^{-3}}$ \\
One sparse row &$ 4.297\cdot 10^{-2}$ & $\mathbf{7.407\cdot 10^{-3}}$ \\
\end{tabular}
\caption{Average KL-distance between predicted and measured probabilities for the different cases. Lower value indicates better fit. \label{table:quant_rate}}
\end{table}

\subsection{Retrieving the key from finite samples}
We now test the layer of security for our protocol given by the fact that $\Phi([Q])$  only is indirectly revealed through samples. 
We here concentrate entirely on the case of $Q$ being dense, since it is more practically relevant. We sample the $h$-vector from a distribution of type \eqref{eq:GaussianFill}, with $z$ having i.i.d. Bernouilli entries ($p=0.5$) and $b$ Gaussian, with diagonal covariance matrix, with diagonal $d=a/\|a\|_\infty$, where the entries of $a \in \C^m$ are absolute values of i.i.d. Gaussian. Note that the latter choice is to ensure that a situation like in \autoref{theo:main_sparse} does not occur. In other words, we make sure that Eve is in a position where \Autoref{algo} has a chance to succeed. When repeating the below experiments with the covariance of $b$ being the identity, the algorithm of course always fails.

\paragraph{\autoref{theo:sample_complex_covy}}First, we validate the $1/\sqrt{r}$-decay rate of the estimation of the vectors $Qx_k$  predicted by  \autoref{theo:sample_complex_covy}. For $n=50$, we plot the relative error (modulo global phase) for different values of $m$ and $r$ in \Autoref{fig:log}. As can be seen, the decay rate is very close to the predicted  $1/\sqrt{r}$-rate.

\paragraph{Full-pipeline recovery}
We now test the empirical recovery performance of the covariance-based attack. That is, given $r$ samples from $M$ different $x_k$,
$$
y_{k,p}=h_p*Qx_k, \qquad k=1, \ldots, M, \;, p =1, \ldots, r,
$$ 
we first use \Autoref{algo} to estimate the values of 
$$
[Qx_k] \in \C^m/\mathbb{S}^1, \qquad k=1, \ldots, M.
$$
We then use those estimates to recover $Q$ by minimizing $\calL$. We stay in the setting $n=50,m=3$, and test two sparsity levels -- $s=4$ and $s=8$ -- and a wide range of sample counts $r$. The results are plotted in \Autoref{fig:finite_r}, with the results for $r=\infty$ from above as a reference. The results reveal that the attack per se works -- and with a growing number of samples, the performance gets closer and closer to the infinite-sample performance.

\begin{figure}
\centering
    \includegraphics[width=.8\linewidth]{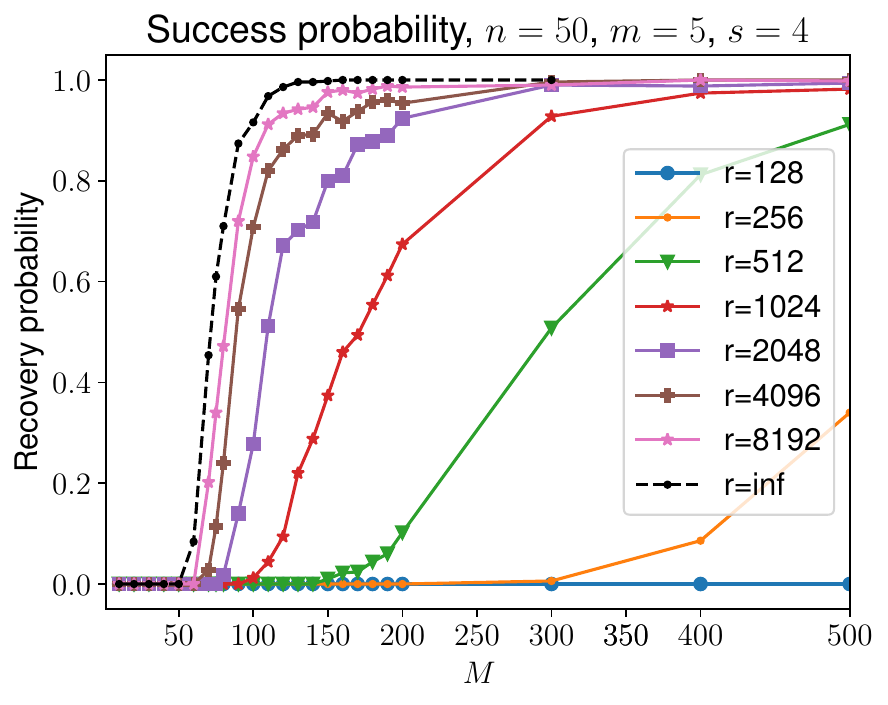} 
    
    \includegraphics[width=.8\linewidth]{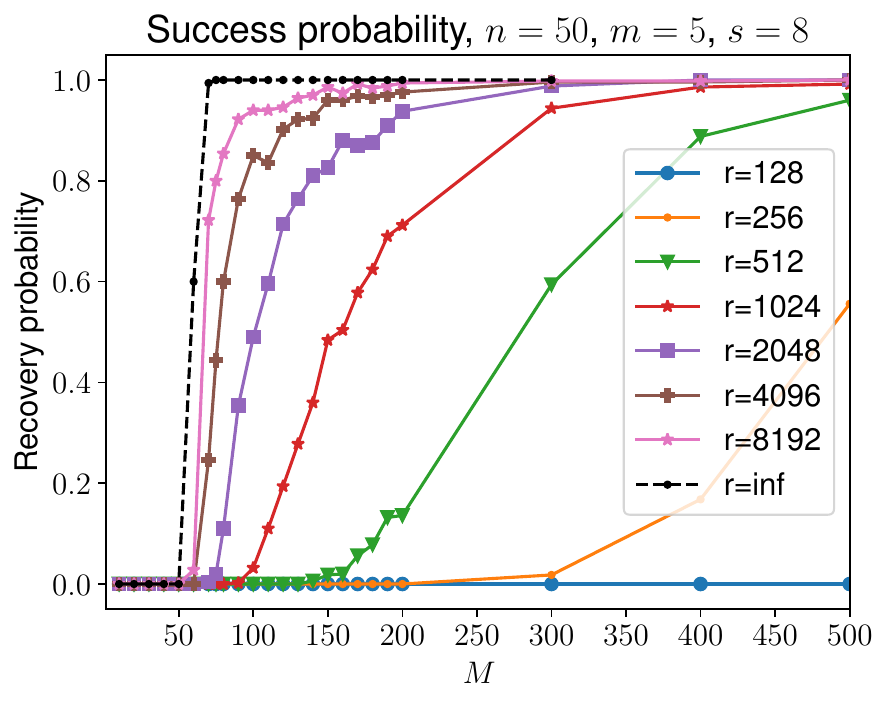}
    \caption{Recovery probabilities for finite number of samples for $n=50,m=5$ and different values of $r$. Top: $s=4$. Bottom: $s=8$.}
    \label{fig:finite_r}
\end{figure}

While the fact that the attack works is important to note, it is equally important to note that it is very expensive. This is even better illustrated by  \Autoref{fig:cost}. Here, we plot the recovery performance against the \emph{total number of samples} $M\cdot r$, where $M$ is the number of distinct plaintexts  and $r$  the number of samples \emph{per $x_k$}). We see that in order to recover the codebook reliably, Eve empirically needs to observe more than a total of $10^5$ samples. This should be compared to the $n=50$ samples she needs to recover $Q$ in the linear case.

\begin{figure}
    \centering
    \includegraphics[width=0.8\linewidth]{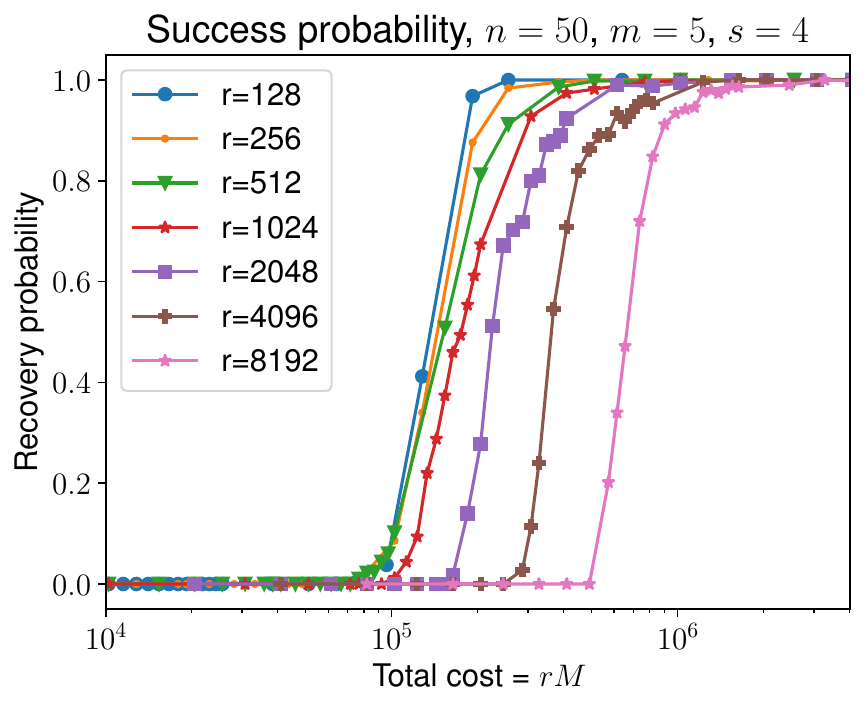}
    \caption{The recovery probability in dependence on the total number of samples $r\cdot M$ for $n=50,m=3$ and $s=4$. $x$-scale is logarithmic.}
    \label{fig:cost}
\end{figure}

\begin{figure}[th]
    \centering
    \includegraphics[width=0.48\textwidth]{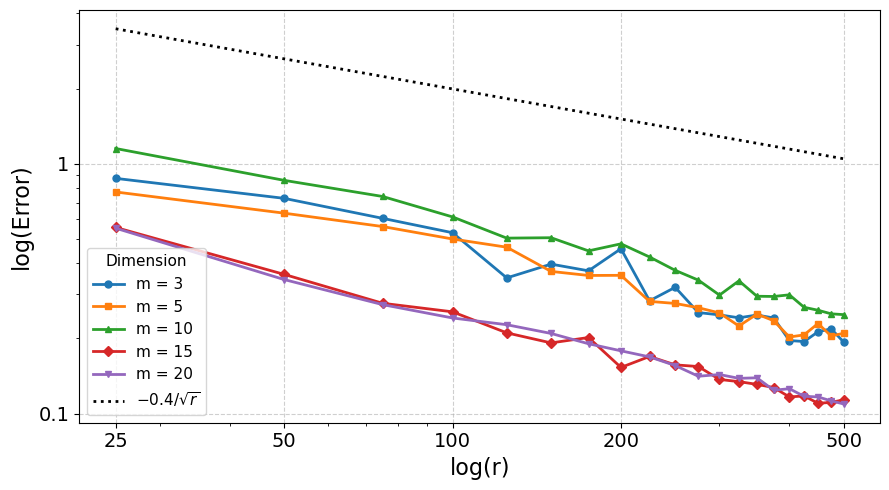}
    \caption{Log-scaled error of estimated $Qx_k$ for varying sample sizes and dimensions.}
    \label{fig:log}
\end{figure}

\section{Discussion}
The experiments reveal a fundamental gap between our theoretical results and the security of BCS in practice. While the injectivity of $\Phi$ is a fundamental source of security of the scheme -- if $\Phi$ is not injective, there are pairs of codebooks $Q$, $Q'$ that cannot be distinguished from an observation of $\Phi([Q])$ -- it is not equivalent to the security of any given codebook. 
For regular, dense $Q$, our experiments instead suggest that a linear cost may suffice. If $Q$ are instead constructed to have sparse rows, the quadratic scaling fit the results better.

Can Alice and Bob use a $Q$ with sparse rows in our pipeline? Not without extra assumptions. Note that when the rows of $Q$ are sparse, $Qx_k=0$ whenever the support of the sparse $x_k$ is disjoint from the support of the rows of $Q$. Bob will have no chance to recover the message $x_k$ from such measurements. However, if all rows in $Q$ have the same support $S$, Alice and Bob can always make sure to send sparse messages that themselves are supported on $S$. They would then effectively obtain a standard blind deconvolution problem of smaller dimension that they are able to solve. The empirical results from above do not directly apply to this modified protocol, at least not when Eve performs a \emph{known plaintext attack}. If she sees the messages Alice sends, they will reveal the set $S$, and Eve's recovery problem will also be reduced to a dense problem with a(n appearant) linear scaling. If she, however, needs to choose the messages herself, i.e., perform a \emph{chosen plaintext attack}, and chooses supports of the messages randomly, \Autoref{fig:recovery_vs_n} gives again a good picture of how her chances of success to recover $Q$ will look like.

One should however note that in any case, our theoretical results and the above discussion always presumes that Eve has access to $\Phi([Q])$. This is a very beneficial for Eve. Indeed, our results reveal that if she applies the covariance-based attack, \Autoref{algo}, she will need a large number of samples to get the same recovery rate as when she is given $\Phi([Q])$ directly. Now, there is of course the possibility that better algorithms than the one we have considered here exist that require fewer samples. Investigating and analyzing such attacks is an interesting direction for future work.

\section{Conclusion}
In this work, we studied the security of Bilinear Compressive Security in the finite-sampling setting. Building upon the previous information-theoretic analysis, we developed the first practical key-recovery attack based on covariance estimation and analyzed its sample complexity. We further investigated the attack experimentally, and showed that, although key recovery from finitely many ciphertext observations is feasible for suitable filter distributions, it requires substantially more observations than the corresponding attack against ordinary linear compressed sensing.

Our experiments also show an interesting gap between the theoretical worst-case identifiability results and the empirical recoverability of particular codebooks. While the theoretical analysis predicts a quadratic dependence on the relevant sparsity parameter, the empirical recovery transition for dense Gaussian codebooks is better described by a linear dependence for the tested parameters. In contrast, the quadratic scaling is only realized when the rows of the codebook $Q$ are sparse.  

Finally, it should be said that the above issue only occurs for the idealized recovery problem, where one assumes that Eve gets access to all $Qx_k \,\mathrm{mod} \, \mathbb{S}$. In practice, these measurements must be estimated from finitely many ciphertext observations. Our covariance-based attack shows that this is possible with the help of covariance estimation under additional suitable assumptions on the filter distribution. However, our experiments indicate that this requires many samples. In future work, we intend to investigate this 'computational security'-aspect more thoroughly by considering more sophisticated algorithms for recovery from finite observations.

\subsection*{Declaration of AI usage}
We used Claude for assistance when writing the code.

\subsection*{Acknowledgement} AF acknowledges support
from the Wallenberg AI, Autonomous Systems and Software Program (WASP) funded by the Knut and
Alice Wallenberg Foundation. HO performed his research during an internship at the department of mathematics and mathematical statistics at Umeå University, and acknowledges support from them. GW and SE were supported by the Federal Ministry for Research, Technology and Space (BMFTR) by the ''EASEPROFIT: Agile and Post-Quantum-Secure Protocols for the Finance Industry'' project, identification number 16KIS2124K, as well as the ''6G-RIC: 6G Research and Innovation Cluster'' project, identification number 16KISK025. GW is also supported by the German Science Foundation (DFG) within priority program SPP 2378: ''ResNets: Resilience in Connected Worlds'' under grant WU 598/12-1.
\bibliographystyle{IEEEtran} 
\bibliography{wdh}

\end{document}